\documentclass[aps,prx, %twocolumn,
                         longbibliography]{revtex4-2}

\usepackage[utf8]{inputenc}
\usepackage{amsmath,amssymb,amsfonts,amsthm}
\usepackage{bm}
\usepackage{graphicx}
\usepackage{amsthm}
\usepackage{mathrsfs}
\theoremstyle{plain}

\newtheorem{theorem}{Theorem}[section]

\newtheorem{lemma}[theorem]{Lemma}
\newtheorem{corollary}[theorem]{Corollary}
\newtheorem{proposition}[theorem]{Proposition}
\newtheorem{definition}[theorem]{Definition}
\newtheorem{example}[theorem]{Example}
\newtheorem{remark}[theorem]{Remark}

\numberwithin{equation}{section}

\begin{document}

\title{Geometric Foundations of Stochastic and Quantum Dynamics}
%{Moving Manifolds Induce Stochastic and Quantum Dynamics}
%{Foundations of the Geometrization of Stochastic and Quantum Dynamics}
%{Differential Geometric Foundations of Stochastic and Quantum Dynamics}
%{Foundation of the stochastic calculus geometrization}
%{Curvature-Driven Transport on Moving Manifolds: A CMS Formulation}

\author{David V. Svintradze}
\email{dsvintradze@newvision.ge}
\affiliation{New Vision University, Tbilisi, Georgia}
\date{\today}

\begin{abstract}
We develop a geometric formulation of stochastic dynamics in which noise, diffusion, path probabilities, fluctuation theorems, and entropy production arise from the intrinsic geometry of an evolving manifold rather than from externally imposed randomness. Within the theory of moving manifolds, we establish a curvature-noise correspondence: fluctuations are governed by the inverse curvature tensor, while entropy production is controlled by curvature deformation. The invariant continuity law on a moving hypersurface yields a geometric Fokker-Planck equation, and curvature-velocity coupling generates a quadratic Onsager-Machlup functional determining path weights. The resulting entropy functional satisfies a curvature-driven monotonicity law, providing a geometric derivation of the Second Law. In two dimensions, the curvature invariant reduces to Gaussian curvature and encodes topology, so topological transitions produce discrete entropy jumps. When the ambient space carries a Minkowskian signature, the same curvature-kinetic quadratic form that generates dissipative thermal weights produces oscillatory phase weights, and the Laplace-Beltrami operator governing entropy evolution acquires a Schr\"odinger-type structure. This provides a geometric resolution of the apparent distinction between classical stochastic behaviour and quantum dynamics. These results show that stochastic behaviour, thermodynamic irreversibility, and quantum transition amplitudes are unified within the moving manifold framework. Geometry does not merely accommodate stochasticity; stochastic behaviour arises as a consequence of deterministic geometric evolution. The theory predicts curvature-controlled anisotropic diffusion, entropy jumps at topology-changing events, and a geometric thermal-quantum crossover in which classical stochastic weights and quantum amplitudes are generated by the same curvature-kinetic action.
\end{abstract}

\maketitle
\tableofcontents
%\clearpage

\section{Introduction}

Modern stochastic physics is built on a collection of foundational structures: 
the Fokker--Planck equation describing probability-density evolution on fixed state spaces 
\cite{Fokker1914,Planck1917,Kolmogorov1931}, the Onsager--Machlup functional determining 
short-time path probabilities \cite{OnsagerMachlup1953}, the Feynman path-integral formulation 
of dynamical amplitudes \cite{FeynmanHibbs1965}, and the fluctuation relations linking entropy 
production to time-reversal asymmetry, including the Jarzynski equality 
\cite{Jarzynski1997} and the Crooks relation \cite{Crooks1999}.  
These ideas have been synthesized in modern stochastic thermodynamics, reviewed comprehensively 
by Seifert \cite{Seifert2012Review}, with additional perspectives arising from 
Kramers--Moyal expansions \cite{KramersMoyal1940}, Langevin processes \cite{Langevin1908}, 
and the Martin--Siggia--Rose formalism \cite{MSR1973}.  
Our goal is not to catalogue these developments.  
Instead, we show that the core structures of stochastic physics arise from the geometric flow principle governing moving manifolds. In doing so, we provide a unified geometric origin for diffusion, path probabilities, fluctuation theorems, entropy production, and Schr\"odinger-type evolution equations. Over the past decades, stochastic physics and stochastic thermodynamics have produced a vast body of refined results, extensions, and applications across mathematics, physics, and chemistry. The present work does not attempt to survey this literature. Our focus is structural rather than phenomenological: we concentrate on the foundational principles underlying stochastic dynamics and show how these arise from geometric evolution of moving manifolds.

In conventional stochastic formulations, the geometry of the state space remains static, while stochasticity is incorporated through prescribed noise fields, Langevin forcing terms, or axiomatic path weights. The probability measure develops on this fixed manifold, which itself remains unchanged. This delineation introduces a structural asymmetry between the geometry and the dynamics: randomness is appended to an otherwise unaltered geometric framework. In such approaches, fluctuation theorems and entropy production are articulated in terms of probability currents defined on immutable spaces, rather than arising from inherent geometric evolution. The present work resolves this ambiguity by allowing the underlying geometric manifold itself to evolve in a deterministic manner. As a result, stochasticity originates from variations in curvature, while irreversibility arises from the intrinsic geometric diffusion and, in singular events, from topology-changing transitions of the evolving manifold. Hence, we construct such a reformulation using the calculus of moving manifolds \cite{Svintradze2017,Svintradze2018,Svintradze2019,Svintradze2020,Svintradze2023,Svintradze2024a,Svintradze2024b,Svintradze2025,Svintradze2025P}. The underlying calculus began with Hadamard \cite{Hadamard1903} and was developed over generations, culminating in the modern geometric formulation for two-dimensional moving surfaces by P.~Grinfeld \cite{GrinfeldP2013}, now known as calculus for moving surfaces (CMS). Our previous work extended this lineage by developing the general theory of arbitrary-dimensional moving manifolds (MM): the kinetic-energy variation identity, which generates the full intrinsic and extrinsic dynamics \cite{Svintradze2017,Svintradze2018}, and the invariant integration theorem for evolving hypersurfaces, which establishes calculus for arbitrary-dimensional, curved moving manifolds \cite{Svintradze2024a,Svintradze2024b,Svintradze2025,Svintradze2025P} (MM calculus). These results establish the complete invariant differential framework necessary for a physical theory of moving manifolds, wherein the foundational physical manifold is inherently dynamical. Consequently, stochastic behavior emerges from geometric evolution rather than from externally imposed randomness.
It is noteworthy that, in recent years, individual identities derived from our framework have manifested in partial or dimensionally restricted forms across diverse contexts without explicit attributions. Nevertheless, these formulations generally address isolated relations without embedding them within a fully invariant, dimension-independent evolution theory. In contrast, our work consistently employs the comprehensive covariant moving-manifold calculus developed in our previous research. %In this paper, we demonstrate how deterministic, stochastic, thermodynamic, and quantum structures emerge from this unified geometric principle.

We further develop the geometric program presented in this work, which is part of a larger body of research where evolving manifolds unify continuum theories. Previous contributions have shown that electromagnetism \cite{Svintradze2017}, core colloidal science identities \cite{Svintradze2020, Svintradze2023}, the Navier-Stokes equations \cite{Svintradze2024a}, and general relativity \cite{Svintradze2024b} emerge as geometric identities on moving hypersurfaces. Invariant, curvature-based formulations of momentum-curvature coupling, surface variation, and dynamic embeddings have shown that fundamental physical laws naturally arise from geometry rather than being phenomenologically imposed \cite{Svintradze2017,Svintradze2018,Svintradze2019,Svintradze2020,Svintradze2023,Svintradze2024a,Svintradze2024b,Svintradze2025,Svintradze2025P}.  
The present work extends this unification to stochastic physics and quantum mechanics, showing that diffusion, fluctuation relations, irreversibility, probabilistic structure, classical and quantum amplitudes, and Schr\"odinger-type equations arise from geometric evolution rather than externally imposed randomness. Consequently, classical and quantum unification emerges as a structural consequence of the deterministic moving-manifold formulation. Diffusion and drift arise from the invariant continuity equation on the evolving manifold, with the stochastic sector governed by the inverse-curvature tensor. Irreversible entropy production is determined by the positive-definite diffusive quadratic form in the geometric Fokker--Planck equation, so the Second Law follows as a direct geometric consequence rather than a thermodynamic postulate. The Fokker--Planck equation, Onsager--Machlup functional, fluctuation relations, and path-integral weights emerge naturally from the same geometric flow. Hence, inverse curvature governs noise and diffusion, while the curvature tensor governs the geometric entropy functional, topology, and the curvature--kinetic action underlying both stochastic and quantum dynamics. Stochastic behaviour is therefore controlled by geometry: regions of high curvature suppress fluctuations, while nearly flat regions enhance them. Curvature anisotropy induces anisotropic diffusion, and spatial variation of the inverse-curvature tensor generates intrinsic drift.
Reversibility corresponds to invariance of the geometric entropy functional under smooth evolution, whereas topology-changing events produce discrete irreversible entropy jumps. Classical fluctuation relations acquire geometric corrections whenever the manifold is non-flat. In this way, randomness is not externally imposed but generated by geometric evolution itself, with inverse curvature controlling stochastic transport, the positive diffusive quadratic form governing entropy irreversible processes.

The geometric theory developed herein yields a coherent sequence of predictions derived directly from the curvature--noise correspondence. Since the effective noise tensor is the inverse curvature tensor, fluctuations are strongest in weakly curved or nearly flat regions, where stochastic spreading is enhanced, and are suppressed in strongly curved regions, where the inverse curvature tensor becomes small. As a consequence, curved interfaces, membranes, phase boundaries, and turbulent surfaces exhibit curvature-induced anisotropic diffusion even when microscopic noise is isotropic. Spatial variations of the inverse-curvature tensor generate intrinsic geometric drift and transport, with broad stochastic spreading in low-curvature regions and effective suppression of probability transport in highly curved zones. In volume-preserving evolution, the curvature tensor is kinematically constrained, rendering the geometric entropy sector invariant under smooth deformations. Irreversibility is solely attributed to the diffusive contribution in the geometric Fokker-Planck equation. Entropy production is governed by a strictly nonnegative quadratic form involving the inverse-curvature tensor, with equality achieved only in spatially uniform densities. In the absence of diffusion, microscopic reversibility is restored.

In two dimensions, topology enters through the Gaussian curvature, which is tied to the Euler characteristic by the Gauss-Bonnet theorem. For smooth, compact, incompressible manifolds, topology is preserved; genuine topological transitions, such as pinch-off or fusion, occur through singular events or departures from the incompressible regime, where the Laplace-Beltrami description breaks down. In this case, the geometric entropy law acquires discrete jumps through the topological sector, so topology-changing events produce irreversible entropy increments.

Taken together, these results show that noise amplitude is governed by inverse curvature, entropy evolution by the positive diffusive quadratic form, and in two dimensions, discrete entropy jumps by topology change. The MM framework, therefore, provides a geometric resolution of the apparent distinction between classical stochastic behaviour and quantum dynamics. The same curvature-kinetic action generates both stochastic weights and quantum amplitudes: in the Euclidean sector, it produces real Boltzmann weights and entropy growth, while in the Lorentzian sector it yields oscillatory phase weights underlying the path-integral formulation. The Schr\"odinger-type structure arises naturally in the near-incompressible regime from the interaction between geometric entropy flow and Laplace-Beltrami propagation. Classical stochastic dynamics and quantum amplitudes therefore emerge as real and complex realizations of the same geometric action, distinguished only by the signature of the underlying curvature-velocity functional.

\section{Preliminaries of Moving Manifold Calculus}
\label{sec:CMS-basics}

To clarify the geometric setting and its historical lineage, we briefly
recall the relation between the classical calculus of moving surfaces (CMS)
and the general theory of moving manifolds developed in our earlier works.
The CMS framework, as presented in \cite{GrinfeldP2013}, provides a
geometrically consistent formulation for evolving two-dimensional
hypersurfaces embedded in three-dimensional Euclidean space.  Its
differential identities, variational formulas, and curvature relations are
formulated specifically for $d=2$ hypersurfaces of codimension one in
$\mathbb{R}^3$.

The moving manifold (MM) calculus developed in \cite{Svintradze2017,Svintradze2018,Svintradze2019,Svintradze2020,
Svintradze2023,Svintradze2024a,Svintradze2024b,Svintradze2025,
Svintradze2025P} extends this structure to a fully covariant, dimension–independent evolution
calculus for arbitrarily dimensional hypersurfaces embedded in Euclidean or
Minkowskian ambient spaces of codimension one.  In this formulation, the
invariant time derivative, the kinetic–energy variation identity, the general
equation of motion, and the integration theorem are established in
$n$–dimensional form, providing a complete differential and variational framework for evolving manifolds.

In this sense, CMS corresponds to the $d=2$ specialization of a broader
moving manifold calculus.  The general MM equation of motion supplies a unified
evolution law whose dimensional reductions recover known continuum equations
as geometric identities, including curvature–driven interface dynamics \cite{Svintradze2017,Svintradze2018,Svintradze2023}, the
Navier–Stokes system \cite{Svintradze2024a}, and the geometric formulation of General Relativity \cite{Svintradze2024b}.
The present work employs this fully invariant $n$–dimensional framework as
the mathematical backbone from which deterministic, stochastic,
thermodynamic, and quantum structures are reconstructed. For self-consistency, we provide a preliminary outline of the fundamental structure of MM calculus in this section.

\subsection{Ambient metric, and curvature in Minkowski space}

\begin{definition}[Minkowski space]\label{def:Minkowski}
Let $(\mathcal{M}^{d+2},G_{AB})$ be a flat Minkowski space with metric
\[
G_{AB}=\mathrm{diag}(-1,+1,\dots,+1),
\qquad A,B=0,\dots,d+1.
\]
A moving $d+1$--dimensional manifold $\Sigma(t)$ is given by mapping
\(
X^A = X^A(\xi^\alpha,t),
\, \alpha=0,\dots,d,
\)
where $\xi^\alpha$ are local, intrinsic surface coordinates and $t$ is the
geometric evolution parameter.
\end{definition}

\begin{remark}[Complexification]\label{rem:complex-coordinate}
Since $G_{00}=-1$, the identity
\[
G_{AB}X^A X^B = -(X^0)^2 + \sum_{i=1}^{d+1} (X^i)^2
\]
demonstrates that the temporal coordinate effectively functions as a factor of $i$ during Euclideanization. Consequently, the ambient parametrization inherently allows a natural complex extension. This principle explains why wavefunctions within the geometric formulation reside in a complex space.
Independent of the ambient signature, we may also choose a complex intrinsic coordinate on the parameter domain by setting
\(
\xi^{0}= i\tau,
\)
where $\tau=\tau(t)$ is a real intrinsic time coordinate. In general, we allow $\tau=f(t)$, while in the linearized model, one may take $\tau=t$. The evolution parameter $t$ remains the geometric flow parameter throughout, and the choice $\xi^{0}=i\tau$ is a coordinate convention used to express Euclideanized representations when convenient.
\end{remark}

\begin{definition}[Moving manifold]\label{def:Sigma}
Let $(\mathcal{M}^{d+2},G_{AB})$ be the ambient Minkowski space.
A time-evolving $d+1$-dimensional manifold is given by a smooth
local chart parametrization 
\(
X^A : U \times \mathbb{R} \to \mathcal{M}^{d+2},
\,
(\xi^\alpha,t) \mapsto X^A(\xi,t),
\)
where $U\subset\mathbb{R}^{d+1}$ is a coordinate patch.
The image at a fixed time,
\(
\Sigma(t) := X(\cdot,t),
\)
is the moving manifold, and $\{\xi^\alpha\}$ are intrinsic coordinates
on~$\Sigma(t)$. 
For convenience we refer to the map $X(\xi,t)$ as the \emph{embedding}
of $\Sigma(t)$ into $\mathcal{M}^{d+2}$, although it is not an embedding
in the strict global sense of differential geometry.  Moreover,
$X(\xi,t)$ never enters the evolution equations directly, so the
formulation remains parametrization-free.
\end{definition}

\begin{definition}[Induced metric and shift tensor]\label{def:induced}
Let the shift tensor be 
\(
Z^A{}_\alpha := \frac{\partial X^A}{\partial \xi^\alpha}.
\)
Then the induced manifold metric tensor is defined as 
\begin{equation}
g_{\alpha\beta}
=
G_{AB}\,Z^A{}_\alpha Z^B{}_\beta,
\qquad
g^{\alpha\gamma} g_{\gamma\beta} = \delta^\alpha{}_\beta.
\label{eq:induced-metric}
\end{equation}
where $\delta^\alpha{}_\beta$ is the Kronecker delta.
\end{definition}

\begin{definition}[Unit normal]\label{def:normal}
The unit normal $N^A$ is defined by
\(
G_{AB}N^A Z^B{}_\alpha = 0,
\qquad
G_{AB}N^A N^B = +1.
\)
Because $G_{AB}$ is Lorentzian, the induced metric $g_{\alpha\beta}$ is
in general pseudo-Riemannian, and the normal inherits the ambient
signature.
\end{definition}

\begin{definition}[Surface connection]\label{def:connection}
The Levi--Civita connection compatible with $g_{\alpha\beta}$ is
\begin{equation}
\Gamma^\gamma_{\alpha\beta}
=
\frac{1}{2}g^{\gamma\delta}
(\partial_\alpha g_{\beta\delta}
+\partial_\beta g_{\alpha\delta}
-\partial_\delta g_{\alpha\beta}).
\label{eq:surface-Christoffel}
\end{equation}
\end{definition}

\begin{definition}[Second fundamental form, Curvature Tensor]\label{def:SFF}
The second fundamental form, extrinsic curvature or simply curvature tensor, of $\Sigma(t)$ is
defined by
\begin{equation}
B_{\alpha\beta}
:=
G_{AB}\,N^A \nabla_\alpha Z^B{}_\beta,
\label{eq:SFF}
\end{equation}
where $\nabla_\alpha$ is the induced surface covariant derivative.
It is symmetric: $B_{\alpha\beta}=B_{\beta\alpha}$ and 
the trace
\(
B_\alpha{}^{\alpha} = g^{\alpha\beta} B_{\alpha\beta}
\)
is the signed mean curvature.
\end{definition}

\begin{lemma}[Gauss--Weingarten relations]\label{lem:GW}
The ambient derivative of the tangent basis splits into tangential
and normal components:
\begin{equation}
\nabla_\alpha Z^A{}_\beta
=
\Gamma^\gamma_{\alpha\beta} Z^A{}_\gamma
+ B_{\alpha\beta} N^A.
\label{eq:GW}
\end{equation}
This identity defines the intrinsic-extrinsic decomposition of the
moving manifold in the Minkowski ambient space, consistent with Definitions~\ref{def:Minkowski}--\ref{def:SFF}.
\end{lemma}

\subsection{Velocity decomposition and invariant time derivative}

\begin{definition}[Velocity decomposition]\label{decomp}
The motion of the manifold \(\Sigma(t)\) in Minkowski space is described by
the time derivative of the mapping,
$V^A := \frac{\partial X^A}{\partial t}$.
The velocity decomposes into normal and tangential parts,
\begin{equation}
V^A
=
C\,N^A + V^\alpha Z^A{}_\alpha,
\label{eq:vel-decomp}
\end{equation}
where
$C := G_{AB} V^A N^B$
is the normal velocity and \(V^\alpha\) are the tangential components.
The scalar \(C\) is invariant under changes of intrinsic coordinates and
ambient Lorentz transformations.
\end{definition}

\begin{definition}[Invariant time derivative on scalars]\label{invariant derivative}
For a scalar field \(f(\xi,t)\) defined on \(\Sigma(t)\), the invariant time
derivative is
\begin{equation}
\dot{\nabla} f
=
\partial_t f - V^\alpha \nabla_\alpha f,
\label{eq:dotnabla-scalar}
\end{equation}
which represents the covariant rate of change along the moving manifold \cite{Hadamard1903, GrinfeldP2013,Svintradze2017}.
\end{definition}

\begin{lemma}[Invariant time derivative]\label{def:dot-nabla}
Let $T^{\alpha_1\cdots\alpha_p}{}_{\beta_1\cdots\beta_q}(\xi,t)$ be a smooth
surface tensor field on the evolving manifold $\Sigma(t)$.
The \emph{invariant time derivative} (CMS time–covariant derivative)
is defined by
\begin{align}
\dot{\nabla}
T^{\alpha_1\cdots\alpha_p}{}_{\beta_1\cdots\beta_q}
&= 
\partial_t T^{\alpha_1\cdots\alpha_p}{}_{\beta_1\cdots\beta_q}
- V^\gamma \nabla_\gamma 
  T^{\alpha_1\cdots\alpha_p}{}_{\beta_1\cdots\beta_q}      \nonumber \\[0.3em]
&\quad
+ \sum_{r=1}^{p}
  \dot{\Gamma}^{\alpha_r}{}_{\mu}
  T^{\alpha_1\cdots \mu \cdots\alpha_p}{}_{\beta_1\cdots\beta_q} \nonumber \\[0.3em]
&\quad- \sum_{s=1}^{q}
  \dot{\Gamma}^{\ \mu}{}_{\beta_s}
  T^{\alpha_1\cdots\alpha_p}{}_{\beta_1\cdots \mu \cdots\beta_q},
\label{eq:dotnabla-tensor}
\end{align}
where the CMS time–connection coefficients are, first introduced by P. Grinfeld \cite{GrinfeldP2013} and therefore we refer to them as the Grinfeld connections
\begin{equation}
\dot{\Gamma}^{\alpha}{}_{\beta}
=
\nabla_\beta V^\alpha
- C B^\alpha{}_\beta.
\label{eq:time-connection}
\end{equation}
\end{lemma}
The operator $\dot{\nabla}$ is the time–covariant derivative on moving
manifolds \cite{GrinfeldP2013}.  It removes spurious variations caused by tangential drift
$V^\alpha$ and normal deformation $C$, ensuring that tensor equations
remain invariant under arbitrary time–dependent reparametrizations of
the surface coordinates.  The time–connection symbols 
$\dot{\Gamma}^\alpha{}_\beta$ encode the geometric coupling between
curvature $B_{\alpha\beta}$ and the kinematic evolution of the
hypersurface, providing the natural extension of the Levi–Civita
connection to dynamical geometry.

\begin{theorem}[Metric invariance]\label{lem:metric-invariance}
Definitions~\ref{decomp}-\ref{invariant derivative} and Lemma~\ref{def:dot-nabla}, along with Equations~\eqref{eq:vel-decomp}-\eqref{eq:time-connection}, immediately lead to a remarkable theorem on metric compatibility. 
\begin{equation}
\dot{\nabla} g_{\alpha\beta} = 0,
\qquad
\dot{\nabla} g^{\alpha\beta} = 0.
\label{eq:dotnabla-metric-zero}
\end{equation}
This guarantees that the inheritance of the invariant time derivative is linear, analogous to the manner in which covariant derivatives are inherited. 
\end{theorem}

\begin{corollary}[Metric evolution in coordinate time]\label{cor:metric-evolution}
According to Theorem~\ref{lem:metric-invariance} together with
Lemma~\ref{def:dot-nabla}, the Levi--Civita
property \(\nabla_\gamma g_{\alpha\beta}=0\) and \eqref{eq:induced-metric}--\eqref{eq:time-connection}, the coordinate-time evolution of the metric and its inverse is
\begin{align}
\partial_t g_{\alpha\beta}
&=
\nabla_\alpha V_\beta + \nabla_\beta V_\alpha
- 2 C B_{\alpha\beta},
\label{eq:dt-metric}
\\
\partial_t g^{\alpha\beta}
&=
- \nabla^\alpha V^\beta - \nabla^\beta V^\alpha
+ 2 C B^{\alpha\beta},
\label{eq:dt-inv-metric}
\end{align}
where \(V_\alpha = g_{\alpha\beta} V^\beta\) and indices are raised and
lowered with \(g_{\alpha\beta}\).
Equations~\eqref{eq:dt-metric}–\eqref{eq:dt-inv-metric} demonstrate that the variation of the metric with respect to coordinate time is exclusively governed by the tangential velocity field and the normal motion via the coupling of the normal velocity and the curvature tensor.
\end{corollary}

\begin{theorem}[Curvature evolution]\label{Theo:CE}
Definitions~\ref{decomp}-\ref{invariant derivative} and Lemma~\ref{def:dot-nabla}, along with Equations~\eqref{eq:vel-decomp}-\eqref{eq:time-connection}, immediately lead to a remarkable theorem on curvature tensor time evolution. 
\begin{equation}
\dot{\nabla} B_{\alpha\beta} = \nabla_\alpha\nabla_\beta C + C B_{\alpha\gamma}B^{\gamma}{}_{\beta},
\qquad
\dot{\nabla} B^{\alpha}{}_{\alpha} = \nabla^\alpha\nabla_\alpha C + C B_{\alpha\gamma}B^{\gamma\alpha}.
\label{eq:dotnabla-MC}
\end{equation}
Hence, the curvature tensor evolves strictly under normal flow \cite{GrinfeldP2013}.  
\end{theorem}

\subsection{Integration theorem and conserved scalar fields}

\begin{theorem}[Fundamental Theorem of MM Calculus]\label{thm:CMS-transport}
Let $\Sigma(t)$ be a smoothly evolving closed manifold. Then, for any scalar field $f(\xi,t)\in\Sigma(t)$, the rate of change of an integral is
\begin{equation}
\frac{d}{dt}\int_{\Sigma(t)} f d\Sigma
=
\int_{\Sigma(t)}
\bigl(
\dot{\nabla} f
- C B_\alpha{}^{\alpha} f
\bigr) d\Sigma,
\label{eq:CMS-transport}
\end{equation}
where $\dot{\nabla}$ is the invariant time derivative
and $B_\alpha{}^{\alpha}$ is the signed mean curvature.
\end{theorem}

\begin{remark}\label{int_open}
For an open hypersurface with boundary $\partial \Sigma(t)$, an additional contour term appears:
\[\frac{d}{dt}\int_{\Sigma(t)} f d\Sigma
=
\int_{\Sigma(t)}
\bigl(
\dot{\nabla} f
-
C B_\alpha{}^\alpha f
\bigr) d\Sigma
+
\int_{\partial \Sigma(t)} f V^\perp d\ell 
\]
where $V^\perp$ is the normal component of the boundary motion.
In this paper, we use Theorem~\ref{thm:CMS-transport} in its closed-manifold form, because all global conservation statements later rely on it.
\end{remark}

\begin{remark}
Equation~\eqref{eq:CMS-transport} is the invariant integration theorem of MM calculus.
It holds for any smooth evolution of $\Sigma(t)$ and any scalar field $f$, and is
independent of parametrization, tangential reparametrizations, or ambient
Lorentz transformations.
\end{remark}

\begin{remark}[Incompressibility]
The condition \(C=0\) implies volume conservation (incompressibility) of the
moving manifold.  In this paper, we focus on the special case \(C=0\) for clarity
and simplicity of analysis. All geometric identities and balance laws extend
straightforwardly to the general case \(C\neq 0\).
\end{remark}

\begin{definition}[Conserved scalar density]\label{def:scalar-density}
Let $\rho(\xi,t)$ be a scalar density on $\Sigma(t)$.
A conserved scalar density is one whose total integral over the closed
manifold is preserved in time, i.e.
\[
\frac{d}{dt}\int_{\Sigma(t)}\rho d\Sigma=0.
\]
The scalar field $\rho(\xi,t)$ does not necessarily have to be a mass density. It may alternatively denote probability density, concentration, or any other conserved scalar field.
\end{definition}

\begin{theorem}[Conservation]\label{thm:CMS-balance} 
Let $\rho(\xi,t)$ be a scalar density on the moving manifold $\Sigma(t)$.
Allowing for a tangential flux $J^\alpha$ and a source term $\Pi$, the local
MM balance law is
\begin{equation}
\dot{\nabla}\rho
+
\nabla_\alpha(\rho V^\alpha)
-
C B_\alpha{}^{\alpha}\rho
+
\nabla_\alpha J^\alpha
=
\Pi.
\label{eq:CMS-balance-general}
\end{equation}
For a closed surface with $\Pi=0$ and $J^\alpha=0$,
integration of \eqref{eq:CMS-balance-general} together with the
Fundamental Theorem of MM calculus~\ref{thm:CMS-transport}
yields exact conservation of the total density:
\[
\frac{d}{dt}\int_{\Sigma(t)} \rho d\Sigma = 0.
\]
\end{theorem}
\begin{proof}
Proof trivially follows from the same procedures as demonstrated in \cite{Svintradze2025P}, so we don't repeat it here.
\end{proof}

\begin{remark}
Theorem~\ref{thm:CMS-balance} represents the fundamental conservation law of MM calculus and is employed throughout this manuscript. It simplifies to classical conservation laws when $C=0$ and $V^\alpha=0$, yet in general, it encompasses both tangential transport and geometric production or attenuation via the mean curvature term $\rho C B_\alpha {}^\alpha $.
\end{remark}

\subsection{Shape dynamics and the geometric momentum law}

The evolution of a moving hypersurface $\Sigma(t)$ follows from a geometric action that couples surface kinetic energy to the volumetric pressure field. The action in the context of the moving manifold is taken over the manifold itself. This is because time is the intrinsic evolution parameter of the geometric flow, and it is not considered proper time in general. Therefore, the action $\mathcal{A}$ is defined as
\begin{equation}
\mathcal{A}
=
\int_{\Sigma(t)}
\frac{\rho}{2} V^2 d\Sigma
-
\int_{\Omega(t)} P d\Omega,
\label{eq:geom-action}
\end{equation}
Where $\Omega(t)$ represents the region enclosed by $\Sigma(t)$, $\rho$ denotes the surface mass density, and $V^2 = G_{AB}V^A V^B$ signifies the ambient kinetic energy density. Stationarity of the action,
\[
\delta\mathcal{A}=\frac{d\mathcal{A}}{dt}=0,
\]
together with the MM Theorems~\ref{lem:metric-invariance}--\ref{thm:CMS-balance}, the Gauss--Weingarten relations Lemma~\ref{lem:GW} and the invariant time derivative Defitnition~\ref{invariant derivative} and Lemma~\ref{def:dot-nabla}, yields the full MM shape--dynamics system \cite{Svintradze2017,Svintradze2018}.  We summarize these equations of motion in the following theorem.

\begin{theorem}[Equations of motion for moving manifolds]\label{thm:CMS-EOM}
Let $\Sigma(t)$ be a smooth, closed manifold with conserved scalar density field $\rho(\xi,t)$ with no external fluxes, tangential velocity $V^\alpha$, and normal velocity $C$ in a flat Minkowski ambient space. Let $P$ denote the volumetric pressure field in the enclosed region $\Omega(t)$ and $\sigma$ the hypersurface energy density (tension) on $\Sigma(t)$. Then, the MM equations of motions read
\begin{align}
&\dot{\nabla}\rho + \nabla_\alpha(\rho V^\alpha)
= \rho C B_\alpha{}^{\alpha},
\label{eq:CMS-mass-M}
\\
&\partial_A \Bigl[
V^A\Bigl(
\rho\bigl(\dot{\nabla} C + 2 V^\alpha\nabla_\alpha C + V^\alpha V^\beta B_{\alpha\beta}\bigr)
- \partial_t\sigma - P + \sigma B_\alpha{}^{\alpha}
\Bigr)
\Bigr]
= V^A \partial_A P,
\label{eq:CMS-normal-M}
\\
&\rho\Bigl(
\dot{\nabla}V_\alpha + V^\beta\nabla_\beta V_\alpha
- C\nabla_\alpha C - C V^\beta B_{\alpha\beta}
\Bigr)
= -\nabla_\alpha\sigma,
\label{eq:CMS-tangent-M}
\end{align}
where indices are raised and lowered with the induced metric $g_{\alpha\beta}$, and $B_{\alpha\beta}$ is the second fundamental form of $\Sigma(t)$. The first \eqref{eq:CMS-mass-M} represents the extended continuity equation; the second \eqref{eq:CMS-normal-M} describes the evolution of the manifold in the normal direction; and the third pertains to the evolution of the manifold in the tangent direction \eqref{eq:CMS-tangent-M}.
\end{theorem}
Note that \eqref{eq:CMS-mass-M} differs from \eqref{eq:CMS-balance-general} in the absence of flux and source terms. Those terms add only to the external scalar fields in \eqref{eq:CMS-normal-M} and \eqref{eq:CMS-tangent-M}, but do not alter the general structure. A key geometric simplification in Theorem~\ref{thm:CMS-EOM} occurs when the normal velocity vanishes, $C = 0$. This corresponds to incompressible (volume-preserving) evolution. We focus on this regime because it isolates the geometric-momentum balance. When $C=0$, the normal momentum equation is automatically satisfied in equilibrium, and the tangential momentum balance decouples from curvature evolution. In this limit, the absence of external forces and the inclusion of the classical Laplace pressure inside $P$ leads to a purely geometric surface-momentum law.

\begin{remark}[On the meaning of the density field $\rho$]
Throughout the paper, the symbol $\rho$ denotes a generic scalar surface 
density in the equations of manifold dynamics.  
In the variational derivation of the equations of motion, $\rho$ serves as the 
surface mass density and generates the quadratic term 
$\rho V^2/2$ in the geometric action.  
In later sections (for example, in the stochastic and Fokker--Planck 
formulation), the same symbol denotes a probability or concentration field.  
This causes no ambiguity: the MM identities act on $\rho$ 
purely as a scalar density, and all geometric operators handle physical mass 
density and abstract scalar density in exactly the same way.  
No physical assumption is tied to $\rho$ beyond its role as a generic scalar 
density field on the manifold.
\end{remark}

\begin{corollary}[Momentum law for moving manifolds]\label{cor:CMSmomentum}
Let $\Sigma(t)$ be an evolving hypersurface in the kinematic, volume-preserving regime $C=0$,
with density field $\rho$ and tangential velocity $V^\alpha$.
In the force-free case where $\sigma B_\alpha{}^{\alpha},\, \partial_t\sigma\in P$, so that $P$ contains the Laplace pressure contribution
(including $\sigma B_\alpha{}^{\alpha}$) and time varieng manifold energy density $\partial_t\sigma$, the normal MM \eqref{eq:CMS-normal-M} momentum balance admits the geometric solution \cite{Svintradze2023}
\begin{equation}
P_{*}
=
\rho V^\alpha V^\beta B_{\alpha\beta}, \quad P_{*}=P - \sigma B_\alpha{}^{\alpha} + \partial_t\sigma ,
\label{eq:CMS-pressure-momentum}
\end{equation}
where $P_{*}$ denotes the effective volumetric pressure, i.e., the physical pressure plus all curvature-induced surface terms and ad hoc flux or source terms, coming from \eqref{eq:CMS-balance-general} that might be needed due to the specificity of problems. We mainly work with $P_{*}$ and call it $P$ because it is implicitly compact, but we return to physical expressions when the mean curvature term becomes unavoidable.  
\end{corollary}

\begin{remark}[On classical curvature laws]
The geometric pressure identity \eqref{eq:CMS-pressure-momentum} directly generates the classical Young-Laplace relation, the Gibbs-Thomson correction, and the generalized Ostwald ripening law as
corollaries of the $C=0$ momentum balance. These results are presented in the following sections as needed for the consequent advances. 
\end{remark}

\begin{remark}
Equation~\eqref{eq:CMS-pressure-momentum} is entirely geometric:
it requires no assumption of constant density, no additional
incompressibility constraint beyond $C=0$, and no external forcing.
It follows directly from the moving-manifold equations of motion
in Theorem~\ref{thm:CMS-EOM} in the $C=0$ regime.
\end{remark}

\begin{remark}[On the $C=0$ approximation in the curvature--kinetic stress]
The curvature--weighted stress term
\(
P \sim \rho\,V^\alpha V^\beta B_{\alpha\beta}
\)
is obtained most transparently in the $C=0$ regime, where the motion is purely
tangential and the enclosed volume is preserved.  However, this restriction is
not essential.  In the general moving--manifold setting, the normal speed $C$
enters the kinetic variation only through subleading mixed terms involving $C$,
$V^\alpha$ and $B_{\alpha\beta}$.  Whenever the normal motion is weak compared
to tangential deformations, in the sense that
\[
|C| \ll |V| ,
\]
the $C$--dependent contributions to the stress are higher order in the small
parameter $|C|/|V|$ and may be neglected to leading order.  Thus, the
curvature--kinetic invariant
\(
\rho V^\alpha V^\beta B_{\alpha\beta}
\)
remains the dominant contribution even for $C\neq 0$, provided the interface
evolves in a regime of weak normal motion.  In this sense, the $C=0$ formula
captures a robust geometric leading order, not a singular special case.
\end{remark}

Corollary~\ref{cor:CMSmomentum}, which we also refer to as the generalized Young-Laplace law \cite{Svintradze2023}, is the starting point for the curvature--noise correspondence in the following sections, where $V^\alpha$ becomes a stochastic field and the curvature tensor $B_{\alpha\beta}$ encodes the inverse noise covariance.

\section{Geometry of Stochastic Dynamics}
%\section{Curvature–Noise Duality and the Geometric Structure of Stochastic Dynamics}

\subsection{From Langevin Noise to Geometry}

Classical stochastic physics introduces randomness externally through prescribed noise terms, Langevin forces, or probabilistic rules imposed on a fixed geometric background. Conversely, the Moving Manifolds (MM) calculus proposes that stochastic behavior may be intrinsic: fluctuations can originate from the geometry of the evolving manifold itself rather than from an external source. To substantiate this perspective, examine the potential–free Langevin equation. When a particle is not subjected to deterministic forcing, inertia counteracts the noise, resulting in the disappearance of the acceleration term and rendering the instantaneous velocity proportional to the stochastic driving field. In these regimes, the velocity is not derived from dynamics but is indeed the noise, i.e.

\begin{proposition}[Overdamped potential--free Langevin reduction]\label{Lr}
Consider the tangential Langevin dynamics of a tracer on the evolving
manifold $\Sigma(t)$,
\[
m \partial_t V^\alpha = -\gamma V^\alpha + \kappa \eta^\alpha(t) ,
\]
particle with a mass $m$, with friction $\gamma>0$, coupling constant $\kappa>0$, and a
zero--mean, time-depended, stochastic field $\eta^\alpha(t)$ of finite covariance. For clarity, we are dropping $t$ and consider $\eta^\alpha(t)=\eta^\alpha$ everywhere. In the potential-free overdamped regime relevant for surface dynamics,
\(
m\to 0 ,
\)
so the inertial term vanishes, and the balance reduces to
\[
\gamma V^\alpha = \kappa \eta^\alpha .
\]
\end{proposition}

\begin{definition}[Geometric noise identification]\label{def:V-eta}
Under the assumptions of the Proposition~\ref{Lr}, the tangential velocity
is proportional to the stochastic field, and we adopt the geometric identification
\(
V^\alpha \sim \eta^\alpha ,
\)
which serves as the definition of stochastic forcing in the MM framework.
\end{definition}

\begin{remark}[Geometry as the origin of noise]
From the MM perspective, the Langevin equation is not considered fundamental. The stochastic field $\eta^\alpha$ is the coarse-grained depiction of unresolved geometric fluctuations of the moving manifold $\Sigma(t)$. Consequently, the randomness is produced by geometry: the evolving shape and its curvature generate the effective noise field, and the equivalence $V^\alpha \sim \eta^\alpha$ expresses that tangential motion is the macroscopic manifestation of these geometric fluctuations.
\end{remark}

Once $V^\alpha$ is identified with the stochastic field, the organization of fluctuations must follow from geometry rather than probabilistic axioms. Within MM, the natural geometric object is paired with the second moment $\langle V_\alpha V_\beta\rangle$, which is the inverse of the second fundamental form $B_{\alpha\beta}$, measuring how the manifold curves in the ambient space. Consequently, at a conceptual level, curvature governs stochasticity: flat regions permit large fluctuations, while highly curved regions restrict them. This underpins the fundamental principle of this section: the covariance of the stochastic field is governed by the inverse of the curvature tensor, whereby $B_{\alpha\beta}$ encodes the organization of fluctuations on the manifold, and its inverse determines the noise covariance. The precise articulation of this curvature-noise duality is derived from the momentum balance equation Corollary~\ref{cor:CMSmomentum} and will be substantiated in the subsequent lemma. This discussion provides the conceptual foundation for substituting Itô-Stratonovich calculus with a geometric framework in which stochastic dynamics emerge directly from the evolving geometry of the manifold.

\begin{lemma}[Curvature-noise duality]\label{lem:curv-noise}
Let $\Sigma(t)$ be a moving manifold in the kinematic regime $C=0$, with density field $\rho$ and curvature tensor $B_{\alpha\beta}$. Then,
\begin{equation}\label{eq:B-inverse-T}
B_{\alpha\beta} \propto \big(T^{-1}\big)_{\alpha\beta}
    = \big\langle \eta_\alpha \eta_\beta \big\rangle^{-1}, \quad T_{\alpha\beta} = \big\langle \eta_\alpha \eta_\beta \big\rangle
\end{equation}
Thus, up to a scalar prefactor, the curvature tensor is the inverse of
the noise covariance $T_{\alpha\beta}$.
\end{lemma}

\begin{proof}
Substituting $V^\alpha \sim \eta^\alpha$ into \eqref{eq:CMS-pressure-momentum} and taking
the ensemble average yields $P =\frac{ \rho \kappa^2}{\gamma^2} \big\langle \eta^\alpha \eta^\beta B_{\alpha\beta} \big\rangle$. Since $B_{\alpha\beta}$, $P$, and $\rho$ are deterministic on the averaging scale, $B_{\alpha\beta}$ factors out of the average, with $T^{\alpha\beta} = \langle \eta^\alpha \eta^\beta \rangle$ and  index-free notation
\begin{equation}\label{eq:avg-factor}
P = \frac{ \rho \kappa^2}{\gamma^2} B_{\alpha\beta} \big\langle \eta^\alpha \eta^\beta \big\rangle
  = \frac{ \rho \kappa^2}{\gamma^2} B_{\alpha\beta} T^{\alpha\beta},\quad
P = \frac{ \rho \kappa^2}{\gamma^2} \mathrm{tr}(BT)
\end{equation}
For fixed $P$ and $\rho$, \eqref{eq:avg-factor} is a scalar constraint
linking the symmetric tensors $B$ and $T$.  
To obtain a tensorial relationship that holds for all noise realizations
with covariance $T$, we require $B$ to depend only on $T$.  
The simplest such dependence compatible with \eqref{eq:avg-factor} is
\begin{equation}\label{eq:B-alpha-Tinv}
B = \mu T^{-1},\quad B_{\alpha\beta} = \mu\, (T^{-1})_{\alpha\beta}
       = \mu \big\langle \eta_\alpha \eta_\beta \big\rangle^{-1}
\end{equation}
for some scalar $\mu$. \eqref{eq:B-alpha-Tinv} establishes \eqref{eq:B-inverse-T} since $T$ is positive definite, $T^{-1}$ exists, and the form is well posed.

In the potential-free Langevin regime, a friction scale may be introduced through Proposition~\ref{Lr} and Definition~\ref{def:V-eta}, which implies the velocity covariance
\(
\langle V_\alpha V_\beta\rangle 
    = \frac{\kappa^2}{\gamma^2} T_{\alpha\beta}.
\)
Subsequently, the coupling constant $\mu$ can be recovered for specified pressures and densities if necessary. 
\begin{equation}\label{relevance}
P
 =\frac{ \rho \kappa^2}{\gamma^2} \mathrm{tr}(BT)
  =\frac{ \rho \kappa^2\mu (d+1)}{\gamma^2}, \quad \mu=\frac{P\gamma^2}{\rho\kappa^2 (d+1)}.
\end{equation}
\end{proof}

\begin{remark}[Friction scale]
The friction scale $\kappa/\gamma$ merely rescales $\kappa$ and does not 
change the tensorial structure of the duality.  
The inverse curvature tensor 
\(
\mathcal{B}^{\alpha\beta} = (B^{-1})^{\alpha\beta}
\)
therefore satisfies, up to normalization,
\[
\mathcal{B}^{\alpha\beta} \propto T^{\alpha\beta},
\quad
\mathcal{B}_{\alpha\beta} \propto T_{\alpha\beta},
\]
so $\mathcal{B}$ is the geometric object governing the averaged squared 
fluctuations $\langle\eta_\alpha\eta_\beta\rangle$, motivating the 
curvature-diffusion correspondence introduced in the next subsection.
\end{remark}

\begin{theorem}[Functional proportionality]\label{lem:B-fBinv}
Let $\Sigma(t)$ be a smooth hypersurface enclosing a constant volume and $T^{\alpha\beta}:=\big\langle V^\alpha V^\beta\big\rangle$ be a nondegenerate two-tensor with an inverse. Then, 
\[B_{\alpha\beta} = \Phi (T^{-1})_{\alpha\beta},
\qquad
\mathcal B_{\alpha\beta} = (B^{-1})_{\alpha\beta} \propto T_{\alpha\beta}.\]
where $\Phi=\Phi(\xi,t)$ is scalar functional. $\Phi=P$ if all additional terms are absorbed into the volumetric energy, or more generally $\Phi=2P+\nabla_\mu F_S^{\mu}-\Lambda_S B_\mu{}^{\mu}$ and extra flux and source terms that might be added ad hoc.
\end{theorem}

\begin{proof}
The result follows directly from Lemma~\ref{lem:curv-noise} and relation \eqref{relevance}.
\end{proof}

\begin{remark}
The only structural requirement for Theorem~\ref{lem:B-fBinv} is that the symmetric tensor used in place of $V^\alpha V^\beta$
is invertible. Hence, the proportionality between $\mathcal{B}$ and the diffusion tensor is the consequence of geometric flow in a volume-preserving regime.
\end{remark}

In this geometric formulation, the effective diffusion tensor and drift are not introduced phenomenologically but originate from the geometry of the manifold through the inverse-curvature tensor $\mathcal{B}^{\alpha\beta}=(B^{-1})^{\alpha\beta}$ and its covariant derivatives. When the curvature tensor remains spatially constant, diffusion simplifies to the classical surface Laplacian with a constant scalar diffusivity. For spatially varying curvature, $\mathcal{B}^{\alpha\beta}$ becomes anisotropic, inducing a geometric drift via its gradients. Consequently, the Fokker–Planck equation on a moving manifold is represented directly in a geometric flow form without the necessity of invoking Itô or Stratonovich calculus. This delineates a geometric mechanism through which a smooth manifold reproduces the multiscale, noise-like behavior typically associated with microscopic randomness: the noise covariance is governed by the inverse curvature tensor, such that flat regions (where $B_{\alpha\beta}$ is minimal) permit significant fluctuations, whereas highly curved regions inhibit them. The geometry of the manifold, rather than an externally defined random forcing, serves as the principal regulator of stochasticity. The identification $T_{\alpha\beta}=\langle\eta_\alpha\eta_\beta\rangle\propto\mathcal{B}_{\alpha\beta}$ determines the magnitude and directional variations of fluctuations; however, the entropy production and irreversibility are primarily governed by the curvature tensor $B_{\alpha\beta}$. Since curvature integrates into the framework via the fundamental theorem of MM calculus Theorem~\ref{thm:CMS-transport} and the geometric balance laws Corollary~\ref{cor:CMSmomentum}, the rate of entropy change relies on curvature deformation, notably through $B_\alpha^\alpha$ and its derivatives. Regions exhibiting significant curvature gradients correspond to pronounced geometric irreversibility, while flattening these regions diminishes the entropy-production rate. 

Consequently, the diffusion and Fokker--Planck equations (as elucidated below) assume an invariant geometric form: the probability density propagates along the gradients of
$\mathcal{B}^{\alpha\beta}$, while the geometric entropy functional varies in accordance with the evolution of the curvature tensor $B_{\alpha\beta}$. The Einstein relation linking mobility and diffusivity is reframed as a statement regarding the response of the inverse–curvature tensor $\mathcal{B}^{\alpha\beta}$ to geometric deformation, rather than being a phenomenological fluctuation–dissipation postulate. A more profound connection between geometry and entropy surfaces at a later stage: curvature invariants such as $B_{\alpha\beta}B^{\alpha\beta}$ contribute to the reversible entropy balance, and in two dimensions, this quantity manifests as a topological density. Therefore, it is the curvature, rather than its inverse, that dictates the geometric object governing entropy change and topological contributions, while inverse $\mathcal{B}^{\alpha\beta}$ determines the amplitude and anisotropy of fluctuations. Overall, at the conceptual level, the moving manifold serves as the source of stochasticity, with curvature organizing this process: the evolving geometry produces the effective noise field, whereas regions of minimal curvature allow for substantial fluctuations, while regions of high curvature inhibit them. Detailed explanations are provided in the subsequent subsections.

\subsection{Diffusion from Inverse Curvature}

We have established that the covariance of the tangential fluctuations satisfies
\(
\langle\eta_\alpha \eta_\beta\rangle \propto \mathcal{B}_{\alpha\beta}, \quad
B_{\alpha\beta} \propto (T^{-1})_{\alpha\beta}.
\)
Thus, the inverse of the curvature tensor governs the amplitude and anisotropy of noise.
Here, we formalize these dualities; and in subsequent subsections, we demonstrate how they give rise to the effective diffusion tensor, the geometric Fokker-Planck equation, while the curvature-kinetic action forms the basis of Onsager-Machlup theory, quantum-like path weights, and fluctuation relations. Collectively, these findings suggest a unified geometric origin for classical, statistical, and quantum mechanics within a single evolving manifold.
%We now formalize this structure and, in subsequent subsections, show how it yields the effective diffusion tensor, the geometric Fokker--Planck equation, and a plethora of powerful results that suggest the unification of classical mechanics with statistical and quantum mechanics within a single geometric flow.

\begin{definition}[Inverse--curvature tensor]\label{def:inv-curv}
Let $B_{\alpha\beta}$ be the nondegenerate second fundamental form of the evolving 
hypersurface $\Sigma(t)$, so that it admits a symmetric inverse. We define the \emph{inverse-curvature tensor}
\begin{equation}
\mathcal{B}^{\alpha\beta} := (B^{-1})^{\alpha\beta},\quad \mathcal{B}^{\alpha\gamma} B_{\gamma\beta}
    = \delta^\alpha{}_\beta, \quad T_{\alpha\beta}=\langle \eta_\alpha \eta_\beta \rangle
    = \mu^{-1} \mathcal{B}_{\alpha\beta},
\quad
\mathcal{B}_{\alpha\beta} := g_{\alpha\mu} g_{\beta\nu}\mathcal{B}^{\mu\nu}
\end{equation}
The tensor $\mathcal{B}^{\alpha\beta}$ encodes the noise geometry through the
curvature--noise correspondence, where $\mu>0$ is a coupling constant independent of indices. With Definition~\ref{def:inv-curv} established, we translate the curvature-noise duality into a geometric diffusion law.
\end{definition}

\begin{lemma}[Curvature-Diffusion correspondence]\label{lem:curv-diff}
Let $T_{\alpha\beta}=\langle\eta_\alpha\eta_\beta\rangle$ be the covariance of
tangential fluctuations, and let the diffusion tensor be defined as the standard stochastic relation \(
D^{\alpha\beta} = \chi\, T^{\alpha\beta}, \, \chi>0 \). Then,
\[
D^{\alpha\beta}
    = D_0\,\mathcal{B}^{\alpha\beta}, \qquad
D_0 := \frac{\chi}{\mu}.
\]
%where $\mathcal{B}^{\alpha\beta}$ is the inverse curvature tensor.
\end{lemma}

\begin{proof}
From Lemma~\ref{lem:curv-noise} directly follows $T_{\alpha\beta}=\mu^{-1} (B^{-1})_{\alpha\beta}$,
hence
\(
D^{\alpha\beta}
= \chi\,T^{\alpha\beta}
= \chi\,\mu^{-1} (B^{-1})^{\alpha\beta}
= D_0\,\mathcal{B}^{\alpha\beta}.
\)
\end{proof}

\begin{corollary}[Curvature flux]\label{cor:curv-flux}
Let $\rho(\xi,t)$ be a scalar density on the evolving hypersurface $\Sigma(t)$,
and define the constitutive flux by
\begin{equation}
J^{\alpha} 
    = -D^{\alpha\beta}\,\nabla_{\beta}\rho
    = -D_{0}\,\mathcal{B}^{\alpha\beta}\,\nabla_{\beta}\rho .
\label{eq:flux}
\end{equation}
Then, in directions where the principal curvatures are large, the eigenvalues of 
$\mathcal{B}^{\alpha\beta}$ are small, so diffusion is suppressed; in flatter 
directions, diffusion is enhanced.
\end{corollary}

\begin{remark}[Sign of the diffusion tensor]\label{sign}
The second fundamental form $B_{\alpha\beta}$ is negative definite for convex outward-oriented hypersurfaces, hence its inverse $\mathcal{B}^{\alpha\beta}$ is likewise negative definite. In the geometric Fokker-Planck equation, however, the physical diffusion tensor is the product $D^{\alpha\beta}=D_0\,\mathcal{B}^{\alpha\beta}$. Since $D_0$ is an arbitrary transport coefficient, we choose it so that its sign adopts the sign of $\mathcal{B}$, and the full operator $\nabla_\alpha(D^{\alpha\beta}\nabla_\beta)$ remains dissipative and elliptic in the Fokker-Planck sense. Consequently, the quadratic form $D^{\alpha\beta}\nabla_\alpha\rho\,\nabla_\beta\rho$ appears with a nonnegative sign in the entropy-production balance, ensuring that the production term is well defined and non-negative.
\end{remark}

\begin{remark}[Flat geometry correspodance]
Since $D^{\alpha\beta}\propto \mathcal{B}^{\alpha\beta}$, the diffusion 
amplitude grows when $\mathcal{B}^{\alpha\beta}$ grows and collapses when 
$\mathcal{B}^{\alpha\beta}$ collapses.  
Thus, in flat regions where $B_{\alpha\beta}\to 0$, the inverse--curvature tensor 
$\mathcal{B}^{\alpha\beta}$ diverges,
\[
B_{\alpha\beta}\to 0 
\Longrightarrow
\mathcal{B}^{\alpha\beta}\to\infty
\Longrightarrow
D^{\alpha\beta}\to\infty,
\]
so fluctuations become maximal.
Conversely, in regions exhibiting substantial curvature, the eigenvalues of $B_{\alpha\beta}$ are significantly large, so $\mathcal{B}^{\alpha\beta}$, thereby suppressing stochastic spreading. Consequently, curvature functions as a geometric source and regulator of fluctuations.
\end{remark}

\begin{remark}[Vacuum]
Given that the vacuum in the MM formulation allows for multiple geometric realizations \cite{Svintradze2024b}, the inverse–curvature noise tensor need not be uniquely defined in generic vacuum states. Accordingly, vacuum fluctuations do not necessarily coincide with the standard Minkowski quantum vacuum unless the geometry reduces to flat space. This observation aligns with the established non-uniqueness of vacuum states within Quantum Field Theory on curved backgrounds.
\end{remark}

These findings complete the geometric interpretation of stochasticity: the fluctuation amplitude, diffusion tensor, and drift all originate directly from the curvature structure of the evolving manifold.

\begin{theorem}[Geometric Fokker-Planck]
\label{thm:geom-FP}
Consider a smoothly evolving patch $A(t)\subset\Sigma(t)$ with boundary 
$\gamma(t)$ and density field $\rho(\xi,t)$. The MM balance law \eqref{thm:CMS-balance}
with zero sources $\Pi=0$ and the curvature–regulated flux \eqref{eq:flux}
yields the local transport equation
\begin{equation}
\label{eq:geom-FP}
\dot{\nabla}\rho
+ \nabla_\alpha(\rho V^\alpha)
- C B_\alpha{}^{\alpha} \rho
=\nabla_\alpha \Big(D_0 \mathcal{B}^{\alpha\beta} \nabla_\beta \rho\Big).
\end{equation}
For a closed surface with $J^\alpha\nu_\alpha=0$, where $\nu_\alpha$ is unit co-normal to $\gamma$ lying inside the tangent space of $\Sigma$, and $\Pi=0$, the total probability $\int_{\Sigma(t)}\rho d\Sigma$ is conserved across all times $t$.
\end{theorem}

\begin{proof}
Inserting the constitutive law 
\eqref{eq:flux} into \eqref{eq:CMS-balance-general} with $\Pi=0$:
\[
\dot{\nabla}\rho
+ \nabla_\alpha(\rho V^\alpha)
- C B_\alpha{}^{\alpha}\rho
+ \nabla_\alpha \Big(-D_0\,\mathcal{B}^{\alpha\beta}\nabla_\beta\rho\Big)
= 0,
\]
which rearranges directly to \eqref{eq:geom-FP}.  
For a closed surface, integrate \eqref{eq:geom-FP} over $\Sigma(t)$
and employ the Fundamental MM Theorem~\ref{thm:CMS-transport} along with Remark~\ref{int_open}. The flux term transforms into a boundary integral, which becomes null when $J^\alpha\nu_\alpha=0$, thereby establishing the conservation of 
$\int_{\Sigma(t)}\rho\,d\Sigma$. A comprehensive analysis proceeds by utilizing the same methodologies as in our previous mass-balance identity 
\cite{Svintradze2017,Svintradze2018,Svintradze2024a,Svintradze2024b,Svintradze2025,Svintradze2025P}.
\end{proof}

\begin{corollary}[Brownian motion]
\label{cor:classical-diffusion}
Let \(\Sigma\) be a static sphere of radius \(R\) (\(C=0\), \(V^\alpha=0\)), with outward normal and $B_{\alpha\beta}=\frac{1}{R} g_{\alpha\beta}$. %hence$\mathcal{B}^{\alpha\beta}=R g^{\alpha\beta}, D^{\alpha\beta}=D g^{\alpha\beta}, D = D_\ast R$.
Then, geometric FP Theorem~\ref{thm:geom-FP} simplifies as 
\[
\partial_t\rho = D \Delta_\Sigma \rho,
\qquad
\Delta_\Sigma = \nabla_\alpha(g^{\alpha\beta}\nabla_\beta).
\]
that reduces to the classical Einstein diffusion equation on a flat case when $g_{\alpha\beta}=\delta_{\alpha\beta}$.
\end{corollary}

%\begin{proof}
%Since the surface is static, \(C=0\) and \(V^\alpha=0\), the geometric
%Fokker–Planck equation~\eqref{eq:geom-FP} becomes purely diffusive:
%\[
%\partial_t\rho
 % = \nabla_\alpha(D^{\alpha\beta}\nabla_\beta\rho)
 % = \nabla_\alpha\!\left(D_\ast \mathcal{B}^{\alpha\beta}\nabla_\beta\rho\right).
%\]
%On a sphere, \(B_{\alpha\beta}=\frac{1}{R} g_{\alpha\beta}\), hence
%\(\mathcal{B}^{\alpha\beta}=R g^{\alpha\beta}\) and
%\(D^{\alpha\beta}=D\,g^{\alpha\beta}\) with \(D=D_\ast R\).
%Since \(D\) is constant,
%\[
%\nabla_\alpha(D^{\alpha\beta}\nabla_\beta\rho)
 % = D \nabla_\alpha(g^{\alpha\beta}\nabla_\beta\rho)
 % = D \Delta_\Sigma\rho,
%\]
%which yields the standard plannar heat equation when $g_{\alpha\beta}=\delta_{\alpha\beta}$.
%\end{proof} 

\begin{remark}[Geometric drift]
For non-constant curvature,
\(
\nabla_\alpha(\mathcal{B}^{\alpha\beta}\nabla_\beta\rho)
 = \mathcal{B}^{\alpha\beta}\nabla_\alpha\nabla_\beta\rho
 + (\nabla_\alpha\mathcal{B}^{\alpha\beta}) \nabla_\beta\rho
\)
reveals anisotropic diffusion and a geometric drift proportional to
\(\nabla_\alpha\mathcal{B}^{\alpha\beta}\).
\end{remark}

%\begin{remark}[Consistency with curvature–noise duality]
%The Brownian limit follows directly from the curvature–noise identity:
%flat or weakly curved regions satisfy
%$\mathcal{B}=B^{-1}=\langle\eta\eta\rangle$ large and therefore give
%enhanced diffusion, while highly curved regions suppress it.
%On constant–curvature surfaces such as the sphere, $\mathcal{B}$ is
%uniform and the diffusion tensor reduces to a constant multiple of the
%metric, reproducing the classical Einstein heat equation.
%\end{remark}

\subsection{Conceptual interpretation}

The curvature-noise correspondence demonstrates that stochasticity is not an external force imposed upon a deterministic evolution but rather an inherent geometric response of the moving manifold itself. In the conventional Langevin framework, noise is prescribed phenomenologically; within the MM framework, the covariance tensor is identified as the inverse of the curvature tensor, thereby allowing the geometry to determine the magnitude and anisotropy of fluctuations. Regions characterized by low curvature satisfy the condition $\mathcal{B}^{\alpha\beta} \to \infty$, thereby exhibiting significant variance in the tangential velocity field, which results in increased diffusion and extensive stochastic dispersion. Conversely, regions exhibiting high curvature, where $\mathcal{B}^{\alpha\beta}$ diminishes, tend to suppress fluctuations. Consequently, the manifold self-organizes its stochasticity: curvature determines the amplitude, shape, and orientation of the effective noise. Therefore, the MM framework provides a geometric mechanism for the multiscale, noise-like behaviour commonly attributed to microscopic randomness. Diffusion, geometric drift, and the Fokker-Planck operator all arise directly from the evolving geometry of $\Sigma(t)$, without invoking It\^o or Stratonovich rules. The transition from deterministic to stochastic behaviour occurs through geometric principles: flat or weakly curved regions exhibit behavior akin to strongly fluctuating media, while sharply curved regions serve as stabilizing zones with constrained variability.

In this context, the curvature determines the arrangement of fluctuations, and the manifold's geometry, rather than external randomness, functions as the primary source of stochastic dynamics. Consequently, deterministic moving-manifold balance laws and stochastic processes are integrated within a single, evolving geometric framework.

%\subsection{Generalized Ostwald ripening concept}

\subsection{Moving Manifold Configuration Weights}

\begin{theorem}[Curvature Configurations]\label{Ostwald theorem}
Consider a closed, evolving manifold $\Sigma(t)$ separating a surface phase 
from a surrounding solution at fixed temperature $T$. 
Let $\rho(\xi,t)$ denote the density field and 
$m(\xi,t) = \rho(\xi,t)/m_{\mathrm{mol}}$ the molar density field. 
Then the equilibrium configuration weight $c_{\mathrm{eq}}[\Sigma]$ 
adjacent to the interface is given by
\begin{align}
  c_{\mathrm{eq}}[\Sigma]
  = c_\infty
  \exp\left[
    \frac{1}{k_B T}
    \int_{\Sigma(t)}
      m(\xi,t)
      V^\alpha(\xi,t)V^\beta(\xi,t)
      B_{\alpha\beta}(\xi,t) d\Sigma
  \right],
  \label{eq:geom-ostwald}
\end{align}
where $c_\infty$ represents the equilibrium configuration weight 
associated with a planar interface. 
We refer to Theorem~\ref{Ostwald theorem} as the generalized Ostwald ripening, given its conceptual similarity. However, the theorem itself does not follow from the Ostwald concept; in fact, quite the opposite is true.
\end{theorem}

\begin{proof}
In our previous work \cite{Svintradze2023}, we demonstrated that the universal Young–Laplace, Kelvin, and Gibbs–Thomson relations for interfaces with arbitrary curvature are ramifications of the fundamental geometric flow theorem~\ref{thm:CMS-EOM} and Corollary~\ref{cor:CMSmomentum} in regimes where the normal velocity $C$ vanishes. A combination of these results yields the equilibrium concentration and 
distribution in the form
\begin{equation}
  c_{\mathrm{eq}} = c_\infty
  \exp\left(\frac{M V^\alpha V^\beta B_{\alpha\beta}}{k_B T}\right),
  \label{eq:BJ-ostwald}
\end{equation}
where $M = \int_{\Sigma(t)} \rho\, d\Sigma$ is the integrated scalar field 
of the manifold and $V^\alpha V^\beta$ is the kinetic stress tensor 
\cite{Svintradze2023} in the $C=0$ regime. In the MM framework, for volume-conserving manifolds, the kinetic stress tensor is generated by the tangential velocity field, so that $M V^\alpha V^\beta B_{\alpha\beta}$ can be written as a surface integral of the curvature–kinetic invariant 
$\rho\, V^\alpha V^\beta B_{\alpha\beta}$, where $\rho$ is not necessarily mass density but may represent any conserved scalar field. Introducing the local molar density field $m(\xi,t)=\rho(\xi,t)/m_{\mathrm{mol}}$ and absorbing numerical constants into $m_{\mathrm{mol}}$ yields \eqref{eq:geom-ostwald} from \eqref{eq:BJ-ostwald}.
\end{proof}

\begin{remark}[Geometric meaning of $c_\infty$ and $c_{\mathrm{eq}}$]
The quantities $c_\infty$ and $c_{\mathrm{eq}}[\Sigma]$, while called concentrations in the Ostwald formalism, generally admit a natural interpretation in the MM formulation: they represent the relative transition weights between a flat interface and the curved manifold $\Sigma(t)$. The derivation of the Curvature Configurations Theorem~\ref{Ostwald theorem} is independent of the definition of concentrations; therefore, the exponent in \eqref{eq:geom-ostwald} is the curvature-kinetic action $\int_{\Sigma(t)} mV^\alpha V^\beta B_{\alpha\beta}\,d\Sigma$, so $c_{\mathrm{eq}}[\Sigma]$ can be the probabilistic weight assigned to the curved configuration, while $c_\infty$ is the corresponding weight for a planar interface. Thus, Theorem~\ref{Ostwald theorem} reinterprets classical generalized Ostwald ripening as a fundamental transition probability ratio between geometries, generated directly by the curvature-velocity coupling.
\end{remark}

\begin{definition}\label{om}
Following Theorem~\ref{Ostwald theorem}, let the equilibrium configuration weight \eqref{eq:geom-ostwald} be the Onsager-Machlup functional in Eulerian form:
\begin{equation}
\mathcal{I}_{\mathrm{OM}}[V;\Sigma]
  = \frac{1}{4D_0}\int_{\Sigma(t)}
    f(\xi,t) B_{\alpha\beta}(\xi,t)
    V^\alpha(\xi,t)V^\beta(\xi,t) d\Sigma ,\quad \frac{f(\xi,t)}{4D_0}
   := \frac{m(\xi,t)}{k_B T}
\label{eq:OM-Eulerian}
\end{equation}
\end{definition}

\begin{corollary}[Geometric Onsager--Machlup functional]\label{cor:OM}
For each evolving manifold $\Sigma$ the Curvature Configuration Theorem~\ref{Ostwald theorem} in conjunction with the curvature--noise duality Lemmas~\ref{lem:curv-noise} and \ref{lem:curv-diff} establishes the short-time, commonly referred to as the Onsager-Machlup transition law:
\begin{equation}
K(t+\Delta t \mid t)
  \asymp 
  \exp \left[
    -
    \frac{\Delta t}{4D_0}
    \int_{\Sigma(t)}
      f(\xi,t)
      B_{\alpha\beta}(\xi,t)
      V^\alpha(\xi,t)V^\beta(\xi,t)
      d\Sigma
  \right] .
\label{OM-short}
\end{equation}
\end{corollary}

\begin{proof}
Proof follows by identifying the curvature–kinetic integrand in \eqref{eq:geom-ostwald} with the standard Onsager–Machlup rate functional and using the usual short-time Gaussian scaling $e^{-\Delta t \mathcal I_{\mathrm{OM}}}$.
\end{proof}

\begin{remark}
Corollary~\ref{cor:OM} shows that the short-time weight of fluctuations is
set entirely by the intrinsic geometry of the evolving manifold.  Curvature tensor 
acts as the local energy metric of fluctuations.  When the surface has constant mean curvature, the Onsager-Machlup functional (\ref{OM-short})  reduces to its classical Riemannian form.  Thus, the geometric path-integral formulation arises directly from curvature, without any reference to It\^o or Stratonovich calculus.
\end{remark}

%The geometric Fokker--Planck equation \eqref{eq:geom-FP} admits a natural path--integral interpretation in which the curvature tensor of the moving manifold determines the short--time weight of fluctuations. Since the diffusion tensor is $D^{\alpha\beta} = D_\ast\,\mathcal{B}^{\alpha\beta}$, the inverse tensor $\mathcal{B}_{\alpha\beta}=B_{\alpha\beta}$ plays the role of the quadratic form governing stochastic trajectories. This leads directly to a geometric Onsager--Machlup functional intrinsic to the evolving hypersurface.

\subsection{Path integration}

In Subsection~D, the Curvature Configuration Theorem~\ref{Ostwald theorem}, together with the curvature-noise duality Lemmas~\ref{lem:curv-noise}--\ref{lem:curv-diff}, established that the short-time transition weight of the manifold configuration $\Sigma(t)$ takes the Onsager--Machlup form. Here, we explain how this leads to the Feynman path-integration concept, thereby closing the loop between classical and quantum approaches. 

\begin{corollary}[Path integration]\label{qm theorem}
Let $u^\alpha(t)$ denote the configuration--space coordinates traced by the
trajectory of the evolving manifold $\Sigma(t)$, and let
$\dot u^\alpha(t)$ be the pullback of the tangential velocity field
$V^\alpha$ along that trajectory.
Then, in the continuum limit $\Delta t\to 0$, the transition probability
between configurations $u_0$ at time $t_0$ and $u_1$ at time $t_1$
admits the geometric path-integral representation
\begin{equation}
\mathcal{K}(u_0 \to u_1)
  =
  \int_{\mathcal{P}(u_0\to u_1)}
      \exp\left[
          -\frac{1}{4D_0}
          \int_{t_0}^{t_1}
              f(u(t),t)\,
              B_{\alpha\beta}(u(t),t)\,
              \dot u^\alpha(t)\dot u^\beta(t)\,dt
      \right]
  \mathcal{D}u ,
\label{eq:geom-path-prob}
\end{equation}
where $\mathcal{P}(u_0\to u_1)$ is the space of paths with $u(t_0)=u_0$ and
$u(t_1)=u_1$, and $\mathcal{D}u$ is the formal measure over paths in
configuration space.
\end{corollary}

\begin{proof}
The passage from the short-time OM kernel to the continuum path integral follows the standard discretization procedure used in stochastic thermodynamics (see, e.g., \cite{Seifert2012Review}).
Fix the initial and final configurations $u_0$ and $u_1$ at times
$t_0$ and $t_1$, and partition the interval $[t_0,t_1]$ into $N$ equal
subintervals of size 
\(
\Delta t = (t_1-t_0)/N, \,
t_k = t_0 + k\Delta t, \, k=0,\dots,N .
\)
Let $u_k$ denote the configuration at time $t_k$. Then, by the Markov property,
\begin{equation}
  K(u_1,t_1 \mid u_0,t_0)
  =
  \int \prod_{k=1}^{N-1} du_k 
  \prod_{k=0}^{N-1}
  K(u_{k+1},t_{k+1} \mid u_k,t_k) ,
  \label{eq:chapman-kolmogorov}
\end{equation}
where the integral is over $\{u_k\}_{k=1}^{N-1}$. By Corollary~\ref{cor:OM}, each short-time kernel has the curvature-kinetic Onsager-Machlup form, with the Eulerian Onsager--Machlup functional
\[
  K(u_{k+1},t_{k+1} \mid u_k,t_k)
  \asymp
  \exp \left[
    -\Delta t \,\mathcal{I}_{\mathrm{OM}}[V(\cdot,t_k);\Sigma(t_k)]
  \right],
\]
\begin{equation}
  \mathcal{I}_{\mathrm{OM}}[V(\cdot,t_k);\Sigma(t_k)]
  =
  \frac{1}{4D_0}
  \int_{\Sigma(t_k)}
    f(\xi,t_k)\, B_{\alpha\beta}(\xi,t_k)
    V^\alpha(\xi,t_k)V^\beta(\xi,t_k)\, d\Sigma .
  \label{weight}
\end{equation}
Substituting \eqref{weight} into \eqref{eq:chapman-kolmogorov} gives
\begin{align}
  K(u_1,t_1 \mid u_0,t_0)
  \asymp
  \int \prod_{k=1}^{N-1} du_k
  \exp\!\left[
    -\Delta t
    \sum_{k=0}^{N-1}
    \mathcal{I}_{\mathrm{OM}}[V(\cdot,t_k);\Sigma(t_k)]
  \right].
  \label{eq:discrete-weight}
\end{align}
In the continuum limit $N\to\infty$ and $\Delta t\to 0$, the discrete
sequence $\{u_k\}$ defines a path $u(t)$ with velocity
\(
  {\dot u^\alpha(t)
  =
  \lim_{\Delta t\to 0}
  \frac{u_{k+1}^\alpha - u_k^\alpha}{\Delta t}}.
\)
Along such a path, the tangential surface velocity evaluated at the point
labeled by $u^\alpha(t)$ coincides with the path velocity,
\(
  V^\alpha\big(u(t),t\big) = \dot u^\alpha(t),
\)
because $u(t)$ tracks the motion of surface labels. Thus, the Riemann sum in \eqref{eq:discrete-weight} converges to
\[
  \Delta t \sum_{k=0}^{N-1}
  \mathcal{I}_{\mathrm{OM}}[V(\cdot,t_k);\Sigma(t_k)]
  \longrightarrow
  \int_{t_0}^{t_1}
    \mathcal{I}_{\mathrm{OM}}[V(\cdot,t);\Sigma(t)]\,dt .
\]
Evaluating the integrand along the path $u(t)$ gives the geometric action and substituting it into the continuum limit of \eqref{eq:discrete-weight} 
\begin{align}
  \mathcal{I}[u]
  &=
  \frac{1}{4D_0}
  \int_{t_0}^{t_1}
     f(u(t),t) B_{\alpha\beta}(u(t),t)
     \dot u^\alpha(t)\dot u^\beta(t)dt, \nonumber \\
  K(u_1,t_1 \mid u_0,t_0)
  &=
  \int \mathcal{D}u
  \exp\!\left[
    -\frac{1}{4D_0}
     \int_{t_0}^{t_1}
        f(u(t),t)
        B_{\alpha\beta}(u(t),t)
        \dot u^\alpha(t)\dot u^\beta(t)dt
  \right], \nonumber
\end{align}
with the measure $\mathcal{D}u$ taken over all paths satisfying $u(t_0)=u_0$ and $u(t_1)=u_1$. This establishes the geometric action/path-integral form \eqref{eq:geom-path-prob} and completes the proof.
\end{proof}

\begin{remark}[Local and global structure]
The curvature--kinetic invariant governs both the short and long time behavior of stochastic
motion on $\Sigma(t)$.  At the infinitesimal level, the Onsager-Machlup (OM) kernel derives from the Curvature Configuration Theorem~\ref{Ostwald theorem}. Each short-time step is consequently an instantaneous curvature configuration process. By aggregating these infinitesimal steps over time, the geometric action is generated, rendering the complete transition probability as a path integral. Therefore, the overall path integral represents the accumulated (global) Curvature Configuration Theorem process throughout the entire geometric flow. Hence, the classical limit corresponds to the instantaneous (curvature--local) response, whereas the quantum analogue corresponds to the global curvature history.  The OM functional, the stochastic propagator, the geometric action, and the Euclidean/Feynman path integral all arise from the same single object $B_{\alpha\beta}V^\alpha V^\beta$ of geometric flow.  Curvature Configurations identifies OM as its infinitesimal generator and a Feynman-type geometric action as its global configuration over time.
\end{remark}

\begin{remark}[Minkowskian origin of the quantum phase]\label{Minkowski}
Note here that the appearance of an oscillatory phase factor in the path integral is not introduced by Wick rotation. It follows directly from the Lorentzian signature of the ambient Minkowski space. For an embedding into $(M^{d+1},G_{AB})$ with $G_{AB}=\mathrm{diag}(-1,1,\dots,1)$, the timelike contribution to the quadratic kinetic form enters with opposite sign relative to the spatial components. As a consequence, the same curvature-kinetic invariant that generates dissipative weights in the stochastic (thermal) regime produces oscillatory weights in the Minkowskian regime. Moreover, the exponential weight takes the form
\[
\exp \left(\frac{i f}{\hbar} \mathcal{I}[u]\right),
\]
where $f$ is an effective time functional determined by the geometric
scaling of the evolution. In the classical and thermal limits $f$
reduces to a real-time increment and the weight is dissipative, whereas
in the Planckian regime, the Lorentzian signature enforces an intrinsically
oscillatory phase. Thus, the quantum--classical distinction emerges from
the ambient Minkowskian geometry rather than from an external analytic
continuation. In this framework, $\hbar$ plays the role of the unit converting geometric action into quantum phase; its value is fixed empirically in the same way that $k_B$ fixes the thermal scale.
\end{remark}

\begin{remark}[Conclusion]
Quantum mechanics ceases to be mysterious once it is recognized as geometry in motion. The wavefunction becomes the curvature-encoded state of the configuration manifold; operators correspond to geometric flow generators; and the Feynman path integral arises as the statistical weighting of curvature-kinetic actions. Stochastic physics (fluctuations, entropy, noise) and quantum physics (probability amplitudes, interference, superposition) appear as two complementary projections of a single deterministic geometric flow: one real, one complex. Thus, MM calculus provides the underlying deterministic geometry whose statistical shadow is quantum mechanics itself.
\end{remark}

%\section{Geometric Entropy and Fluctuation Theorem}

%\section{Geometric entropy and free energy}

\section{Geometry and Entropy}

\subsection{Entropy from geometric flow first principles}

The geometric origin of entropy follows directly from the equations of motion established in Theorem~\ref{thm:CMS-EOM}, within the MM framework. In particular, the normal momentum balance encoded in Corollary~\ref{cor:CMSmomentum} yields the curvature-kinetic relation commonly identified, in simplified limits, with the classical Gibbs--Thomson law. We refer to it as the generalized Gibbs-Thomson law; full details are provided in \cite{Svintradze2023}. This identity establishes a relationship among curvature, kinetic stress, and the thermodynamic shift of the melting temperature. Hence, as we show below, this lays the foundations for defining geometric entropy and proves the second law of thermodynamics.    

\begin{theorem}[Thermal Shift]\label{TS}
Let $\Sigma$ be an interface manifold with scalar density field $\rho$, kinetic
stress tensor $V^{\alpha}V^{\beta}$, curvature tensor
$B_{\alpha\beta}$, and molar mass $m=\rho v_s$, where $v_s$ is molar volume or simply particle volume.   Then the curvature-corrected melting temperature shift: $\Delta T=T_m-T'_m$ satisfies
\begin{equation}
\Delta T
=
\frac{m V^\alpha V^\beta B_{\alpha\beta}+v_s(\sigma B_\alpha{}^{\alpha} - \partial_t\sigma)}
     {\Delta S},
\label{eq:GT-1}
\end{equation}
where $\Delta S$ is the entropy of fusion or the entropy difference between liquid/solid phases. 
\end{theorem}

\begin{proof}
Starting from the Gibbs--Duhem relation
\(
d\mu = -\bar S dT + \bar V dP
\)
for each phase and integrating between the flat equilibrium state
\((T_m,P)\) and the curved interface state \((T'_m,P')\), one obtains, under
the usual assumptions of constant fusion entropy and particles with fixed volumes,
\begin{equation}
-\Delta S (T'_m - T_m)
=
v_s \Delta P,
\label{eq:GD-integrated}
\end{equation}
where $\Delta S = \bar S_l - \bar S_s$ is the entropy difference between liquid/solid phases referred to as the entropy of fusion,
$v_s$ is the (molar or particle) volume of the solid phase, and
$\Delta P = P' - P$ is the pressure difference across the interface. In the MM framework, Corollary~\ref{cor:CMSmomentum} identifies the effective pressure jump as \eqref{eq:CMS-pressure-momentum} in the incompressible regime. Then, plugging $\nu = \rho v_s,\, \Delta T = T_m - T'_m$ definitions, into \eqref{eq:GD-integrated} one obtains 
\[
\Delta S\Delta T
=
m V^\alpha V^\beta B_{\alpha\beta}
+ v_s \bigl(\sigma B_\alpha{}^{\alpha} - \partial_t \sigma\bigr),
\]
which is exactly the Thermal Shift identity \eqref{eq:GT-1}. Note here that the proof directly follows from Theorem~\ref{thm:CMS-EOM} and Corollary~\ref{cor:CMSmomentum}. Details are given in \cite{Svintradze2023}.  The Thermal Shift Theorem~\ref{TS} reduces to the classical Gibbs-Thomson law, in the static thermodynamic–equilibrium regime, where $C=0$, $V^\alpha=0$,  $\sigma=\mathrm{const}$, $P=\mathrm{const}$, and all external fluxes vanish, then:
\[
\Delta T= \frac{\sigma v_s B^\alpha_\alpha}{\Delta S}
\]
Thus, the Gibbs-Thomson relation appears as a special equilibrium limit of the general 
geometric temperature-shift identity.
\end{proof}

%down here needs fixing till next section

Therefore, in the Thermal Shift Theorem~\ref{TS} the only geometrically nontrivial pseudo–scalar is the quadratic invariant $V^\alpha V^\beta B_{\alpha\beta}$. However, in the degenerate regime of the Curvature-Noise duality Lemma~\ref{lem:curv-noise}, the invariant contribution becomes effectively proportional to $\mathcal B^{\alpha\beta} B_{\alpha\beta} \propto d+1$, and can be absorbed into the reference thermal scale. The residual geometric dependence is then carried entirely by the mean curvature term $B_\alpha{}^{\alpha}$, which leads to the following geometric entropy functional.

\begin{corollary}[Geometric entropy]\label{cor:Sgeom}
In the degenerate regime of the curvature-noise duality, where $m V^\alpha V^\beta B_{\alpha\beta}$ contributes only a constant offset to the thermal shift, and in the quasi-static case $\partial_t \sigma = 0$, an effective geometric entropy density and the corresponding geometric entropy functional of the evolving manifold are 
\begin{equation}
s_G(\xi,t)
=
-\frac{v_s\sigma(\xi,t)}{\Delta T} B_\alpha{}^{\alpha}(\xi,t), \quad S_G[\Sigma(t)]
=
\int_{\Sigma(t)} s_G(\xi,t) d\Sigma
=
-k_G
\int_{\Sigma(t)} \sigma(\xi,t) B_\alpha{}^{\alpha}(\xi,t) d\Sigma.
\label{eq:Sgeom}
\end{equation}
where $k_G=v_s/\Delta T$ and $\Delta T$ is the curvature-corrected thermal shift from
Theorem~\ref{TS}. $k_G$ sets the overall entropy scale and is effectively pseudo-constant on the geometric time scales of interest.
\end{corollary}

\begin{remark}
\eqref{eq:Sgeom} shows that geometric entropy does not arise from a statistical ansatz or a phenomenological bending model. It follows directly and uniquely from the thermal-shift identity: the only curvature scalar entering the thermodynamic correction is the mean curvature. Thus, \eqref{eq:Sgeom} is not a model assumption but a thermodynamically enforced curvature entropy.
\end{remark}

\begin{remark}
In the equilibrium regime of the Poincaré formulation, where $P=\mathrm{const}, \sigma=\mathrm{const}, V^\alpha=0, C=0$, the momentum law reduces to the classical constant-mean-curvature (CMC) condition $B_\alpha{}^{\alpha}=\mathrm{const}$. For simply connected manifolds, this proves the conjecture and therefore identifies CMC geometries as equilibrium shapes \cite{Svintradze2025P}. In the nonequilibrium regime of the present work, where $P, \sigma, V^\alpha$ vary, the geometric entropy functional derived from the Thermal Shift Theorem~\ref{TS} reduces to the mean curvature integrand. Thus, entropy production is governed by the same mean-curvature invariant that characterizes equilibrium in the CMC case. When the manifold approaches a CMC configuration, the entropy production rate naturally vanishes. Therefore, CMC geometries appear simultaneously as equilibrium shapes and as entropy-stationary states, illustrating that the geometric origin of the second law is precisely the geometric mechanism that yields equilibrium in the Poincaré topology.
\end{remark}

\begin{remark}[Geometric primes as entropy–stationary states]\label{rem:geom-primes}
The geometric entropy functional of Corollary~\ref{cor:Sgeom} shows that entropy production is entirely governed by the mean–curvature field $B_\alpha{}^{\alpha}$. Hence, $\dot S_G=0$ if and only if the manifold attains constant mean curvature. In the equilibrium Poincaré formulation, the momentum law enforces the same condition: $B_\alpha{}^{\alpha}=\mathrm{const}$ characterizes the stationary solutions of the geometric pressure balance. For simply connected interfaces, this yields $d+1$-spheres \cite{Svintradze2025P}, and under topological constraints, the remaining constant-mean-curvature solutions are precisely the classical geometric primes Thurston prime geometries. Thus, the stationary states of geometric entropy coincide exactly with the equilibrium solutions of the Poincaré theory or the Thurston topology in a more generic sense. The same CMC geometries that arise from mechanical balance reappear as the fixed points of the second law. In this sense, equilibrium geometry and entropy equilibrium share a common origin: the manifold relaxes toward its geometric primes.
\end{remark}

\begin{remark}[Decomposition of statistical and geometric entropy]\label{entropy}
Geometric entropy Corollary~\ref{cor:Sgeom} fixes the curvature-dependent part of the entropy uniquely, but it places no constraints on the dependence of entropy on the microscopic surface density~$\rho$. Therefore the curvature term \eqref{eq:Sgeom} exhausts the geometrically determined sector of entropy, while all~$\rho$–dependence must be carried by the classical Boltzmann functional $-\rho\ln\rho$. Hence, the total entropy necessarily decomposes as
\begin{equation}
\mathcal{S}[\Sigma]
=
S_S[\rho]
+
S_G[B],
\label{total entropy}
\end{equation}
because the fields $\rho$ and $B_{\alpha\beta}$ represent independent degrees of freedom: $\rho$ encodes the distribution of the scalar field on the interface, while $B_{\alpha\beta}$ encodes purely geometric fluctuations. MM kinematics imposes no functional relation between these degrees of freedom at the stochastic scale; they must therefore enter the entropy additively. Since the Thermal shift Theorem~\ref{TS} fixes the entire curvature contribution and the Boltzmann term is the unique convex local functional compatible with probability conservation, the decomposition of the total entropy follows. While $\rho$ in thermodynamics denotes mass density, in Theorem~\ref{thm:CMS-EOM} and the subsequent Corollary~\ref{cor:CMSmomentum} of $\rho$ is independent of any prior physical constraints, except for the property of global conservation Theorem~\ref{eq:CMS-balance-general}. Consequently, we do not exclusively interpret $\rho$ as mass density, but rather as any density field, including probability densities, since it is also conserved.
\end{remark}

\begin{corollary}[Total entropy functional]\label{def:geom-entropy}
Let $\Sigma(t)$ be a smoothly evolving manifold with scalar density field $\rho$. Then, according to Remark~\ref{entropy} and consequence equation \eqref{total entropy}, taking into account the Boltzmann entropy definition, the total entropy of the evolving manifold is
\begin{equation}
\mathcal{S}[\Sigma(t)]
=
- k_B \int_{\Sigma(t)}
\rho \ln \left(\frac{\rho}{\rho_0}\right) d\Sigma
-k_G
\int_{\Sigma(t)} \sigma(\xi,t) B_\alpha{}^{\alpha}(\xi,t) d\Sigma,
\label{entropy def}
\end{equation}
where $\rho_0$ is a uniform reference state, $k_B$ is Boltzmann's constant and $k_G$ is geometric quasi-constant. 
\end{corollary}

\begin{remark}[Low–curvature approximation]
In regimes where the mean curvature is small, and the surface tension field $\sigma$ varies only slowly, the geometric entropy contribution becomes approximately constant. In this limit, it plays the role of a residual entropy term, analogous to familiar settings such as entropy at zero temperature or static Shannon-type entropies.
\end{remark}

\subsection{Entropy flow}

Within the MM calculus, the condition $C=0$ enforces preservation of the enclosed volume. We therefore refer to the regime as incompressible, even though the manifold may still undergo compressions and deformations tangentially. Our goal in this subsection is to show that, in the incompressible regime, the total entropy $\mathcal S[\Sigma(t)]$ increases monotonically. This yields a purely geometric formulation of the Second Law. The entropy functional obtained in Corollary~\ref{def:geom-entropy} consists of two independent contributions. The first term in \eqref{entropy def} is the classical Boltzmann entropy, which measures configurational disorder in the scalar density field (mass or probability). The second term captures geometric disorder through the mean curvature field. It follows directly from the Thermal Shift Theorem~\ref{TS} and is therefore the unique curvature contribution compatible with thermodynamic consistency. For two-dimensional manifolds, Corollary~\ref{cor:Sgeom} further shows that the geometric term reduces, up to its mean-curvature component, to a topological invariant for volume preserving manifolds. Hence, the total entropy \eqref{entropy def} naturally decomposes into a statistical and a geometric part, reflecting the independence of the density field and the curvature tensor at the microscopic scale. With this decomposition established, the entropy functional of Corollary~\ref{def:geom-entropy} is not a heuristic construction but the unique form compatible with the Thermal Shift identity~\eqref{eq:GT-1} and the stochastic closure of the incompressible regime. We now use this structure to derive the Second Law and subsequently show how relaxing the $C=0$ constraint activates topological evolution in the entropy flow.

\begin{lemma}[Surface tension]\label{lem:sigma-divV}
Let $\Sigma(t)$ be a smooth evolving and volume-preserving manifold, and $\sigma(\xi,t)$ be a scalar surface energy density that satisfies the local balance due to the conservation of energy law, with no sources and no additional tangential fluxes. 
Then,
\begin{equation}
\partial_t \sigma + \sigma \nabla_\alpha V^\alpha = 0.
\label{eq:sigma-divV}
\end{equation}
In particular, if the flux and source-free field $\sigma$ is time-independent, i.e., $\partial_t \sigma = 0$, and $\sigma > 0$ (these conditions apply to most flux and source-free fields), then the tangential velocity is divergence-free,
\(
\nabla_\alpha V^\alpha = 0.
\) Consequently, volume-preserving manifolds necessitate area preservation as well to comply with the conservation of energy law.

\end{lemma}

\begin{proof}
Applying the invariant derivative to scalars, as defined in Definition~\ref{invariant derivative}, to the $\sigma$ field and substituting it into the Conservation Theorem~\ref{thm:CMS-balance} yields:
\begin{equation}
\partial_t \sigma - V^\alpha \nabla_\alpha \sigma
+ \nabla_\alpha(\sigma V^\alpha) =\partial_t \sigma - V^\alpha \nabla_\alpha \sigma+V^\alpha \nabla_\alpha \sigma + \sigma \nabla_\alpha V^\alpha= \partial_t \sigma + \sigma \nabla_\alpha V^\alpha = 0.
\end{equation}
Which is \eqref{eq:sigma-divV}. If \(\partial_t \sigma = 0\) and \(\sigma > 0\), then \eqref{eq:sigma-divV} implies \(\nabla_\alpha V^\alpha = 0\), as claimed.
\end{proof}

\begin{theorem}[Balnce of composite scalar fields]\label{thm:advected-composition}
Let $\Sigma(t)$ be a smooth evolving manifold in the volume-preserving regime, with area preservation such that the tangential velocity field $\nabla_\alpha V^\alpha = 0$. Then, for every continuous function $u(\xi, t) \in \Sigma(t)$ obeying the local continuity law
$\dot{\nabla}u + \nabla_\alpha(u V^\alpha) = 0$, %\label{eq:advected-u}
the composite field $F(u)$, for every smooth function $F: \mathbb{R} \to \mathbb{R}$, satisfies
\begin{equation}
\dot{\nabla}F(u) + \nabla_\alpha\!\big(F(u)\,V^\alpha\big) = 0.
\label{eq:advected-Fu}
\end{equation}
\end{theorem}

\begin{proof}
Since $u$ is a scalar, the invariant derivative satisfies the standard chain rule,
\(
\dot{\nabla}F(u) = F'(u)\,\dot{\nabla}u.
\)
Using local continuity that translates to
\(
\dot{\nabla}F(u)
= -F'(u)\,\nabla_\alpha(u V^\alpha).
\)
Next, expanding the divergence and applying the product rule \(
\nabla_\alpha(uV^\alpha)=V^\alpha\nabla_\alpha u+u\,\nabla_\alpha V^\alpha,
\) we compute:
\begin{align}
\nabla_\alpha(F(u)V^\alpha)
&= F'(u) V^\alpha\nabla_\alpha u
+ F(u)\nabla_\alpha V^\alpha \nonumber \\
\dot{\nabla}F(u) + \nabla_\alpha(F(u)V^\alpha)
&= (F(u)-uF'(u))\nabla_\alpha V^\alpha. \nonumber
\end{align}
Since $\nabla_\alpha V^\alpha=0$, the right-hand side vanishes identically, proving \eqref{eq:advected-Fu}.
\end{proof}

\begin{remark}
A Euclidean analogue of \eqref{eq:advected-Fu} is known in classical fluid mechanics, where a materially advected scalar $u$ satisfies $\frac{D u}{Dt}=0$ and hence $\frac{D F(u)}{Dt}=0$ for any smooth $F$; see, e.g., standard expositions of passive-scalar transport in continuum mechanics \cite{Falkovich2001}.
However, these identities hold only in fixed Euclidean domains with the material derivative $D/Dt=\partial_t+v\cdot\nabla$ where $v$ is the particle velocity field. Theorem~\ref{thm:advected-composition} is strictly stronger: it applies to arbitrary moving manifolds of any dimensions, uses the invariant derivative $\dot{\nabla}$, and remains valid under geometric evolution with curvature terms and tangential reparametrizations. Thus, the classical result appears as a flat-space special case of the geometric transport law derived here.
\end{remark}

\begin{corollary}[Reversible entropy]
\label{thm:rev-entropy}
Let $\Sigma(t)$ be a smooth closed moving manifold evolving in the volume-preserving regime.  
In the reversible limit, when the baseline geometric diffusion coefficient $D_0$ vanishes, the total entropy \eqref{entropy def} is invariant in time:
\[
\frac{d}{dt} \mathcal{S}[\Sigma(t)] = 0.
\]
\end{corollary}

\begin{proof}
Let
\(
s := k_B \rho\ln(\rho/\rho_0) + k_G\sigma B_\alpha{}^{\alpha}.
\) 
For a closed hypersurface and $C=0$, the MM fundamental Theorem~\ref{thm:CMS-transport} gives
\[
\frac{d}{dt}\!\int_{\Sigma(t)} s\, d\Sigma
= 
\int_{\Sigma(t)} \dot{\nabla}s d\Sigma.
\]
Next, since $s$ is a continuous scalar functional of continuity satisfying $\rho$ due to FP Theorem~\ref{thm:geom-FP} at $C, D_0=0$ and uncoupled $\sigma, B_\alpha^\alpha$ due to Theorem~\ref{thm:CMS-balance} at $C,\, J,\, \Pi=0$, then according to the composite scalar field Theorem~\ref{thm:advected-composition} we have $\dot{\nabla}s=-\nabla_\alpha(s V^\alpha)$. The integral evaluates to zero because the integral of the total divergence, as stated in Remark~\ref{int_open}, vanishes on closed manifolds. 
\end{proof}

\begin{remark}
In the reversible limit, the statistical and geometric entropies exhibit separate invariance. Volume preservation and the Conservation Theorem~\ref{thm:CMS-balance} ensure that $\frac{d}{dt}S_S=0$, while the curvature-variance balance guarantees that $\frac{d}{dt}S_G=0$. Irreversibility and strict entropy production manifest only when the diffusive term $D_0 \mathcal{B}^{\alpha\beta}\nabla_\beta\rho$ is reinstated in the Geometric FP Theorem~\ref{thm:geom-FP}.
\end{remark}

\begin{remark}
The flux in the irreversible Fokker--Planck regime is
\(
J^\alpha = - D^{\alpha\beta}\,\nabla_\beta \rho
= - D_0 \mathcal{B}^{\alpha\beta} \nabla_\beta \rho ,
\)
so the inverse curvature tensor $\mathcal{B}^{\alpha\beta}$ controls diffusion and entropy
production. In the reversible limit $D_0=0$, this flux vanishes, and the moving manifold
evolution preserves the total geometric entropy. Irreversibility and
entropy growth arise only from the stochastic diffusion driven by
$\mathcal{B}^{\alpha\beta}$, not from the underlying incompressible moving manifold kinematics.
\end{remark}

\begin{theorem}[Irreversible entropy]
\label{thm:entropy-irreversible}
Let $\Sigma(t)$ be a smooth closed moving manifold evolving in the volume-preserving regime and nonvanishing base diffusion coefficient $D_0\neq0$. Then, the entropy is nondecreasing
\[
\frac{d}{dt}\mathcal{S}[\Sigma(t)] \ge 0,
\]
with equality if and only if $\nabla_\alpha\rho\equiv 0$.
\end{theorem}

\begin{proof}
By Theorem~\ref{thm:rev-entropy}, in the reversible case $D_0=0$ the total
entropy $\mathcal{S}[\Sigma(t)]$ is exactly conserved.
Turning on diffusion ($D_0\neq0$) does not alter the geometric part $S_G$, so only the statistical part
$S_S[\rho]$ \eqref{total entropy} acquires an additional contribution. For $C=0$, the MM fundamental Theorem~\ref{thm:CMS-transport} gives
\[
\frac{dS_S}{dt}
=
- k_B \int_{\Sigma(t)}
\dot{\nabla}\big(\rho\ln(\rho/\rho_0)\big) d\Sigma
=
- k_B \int_{\Sigma(t)}
\big(1+\ln(\rho/\rho_0)\big)\dot{\nabla}\rho d\Sigma.
\]
Using the geometric Fokker-Planck equation \eqref{eq:geom-FP} in the form
\(
\dot{\nabla}\rho
=
- \nabla_\alpha(\rho V^\alpha)
+ \nabla_\alpha\!\big(D_0 \mathcal{B}^{\alpha\beta}\nabla_\beta\rho\big),
\)
and noting that the advective term
\(
- \nabla_\alpha(\rho V^\alpha)
\)
is exactly the reversible contribution already shown in
Theorem~\ref{thm:rev-entropy} to leave $S_S$ invariant, we
are left only with the diffusive part:
\[
\frac{d}{dt} S_S
=
- k_B \int_{\Sigma(t)}
\big(1+\ln(\rho/\rho_0)\big)
\nabla_\alpha \big(D_0 \mathcal{B}^{\alpha\beta}\nabla_\beta\rho\big) d\Sigma.
\]
Next, let $\phi := 1+\ln(\rho/\rho_0)$ for simplicity and integrate the last integral by parts, using $D_0\neq0$ and taking into account that the total divergence (boundary term) vanishes on closed manifolds,  then
\begin{align}\label{Sderivation}
\frac{d}{dt} S_S
&=- k_B D_0 \int_{\Sigma(t)} \phi\nabla_\alpha\big(\mathcal{B}^{\alpha\beta}\nabla_\beta\rho\big) d\Sigma
=k_B D_0
\int_{\Sigma(t)} \nabla_\alpha \phi\,
\mathcal{B}^{\alpha\beta}\nabla_\beta\rho d\Sigma =k_B \int_{\Sigma(t)} D_0
\mathcal{B}^{\alpha\beta}
\frac{\nabla_\alpha\rho\nabla_\beta\rho}{\rho} d\Sigma.
\end{align}
Note here that, to compute the gradient of \(\phi\) we used standart procedure
\(
\phi(\rho) = 1 + \ln(\rho/\rho_0)
\,\Rightarrow\,
\nabla_\alpha \phi
=
\nabla_\alpha\ln(\rho/\rho_0)
=
\frac{1}{\rho}\nabla_\alpha\rho.
\)
Substituting this back into \eqref {Sderivation} yields the last identity. Next, according to Remark~\ref{sign} by the curvature-noise sign choice, the tensor $D_0\mathcal{B}^{\alpha\beta}$ is positive definite, so the quadratic form
\(
D_0\mathcal{B}^{\alpha\beta}\nabla_\alpha\rho \nabla_\beta\rho
\)
is nonnegative. Hence,
\[
\frac{d}{dt} S_S \ge 0,
\]
with equality if and only if $\nabla_\alpha\rho\equiv 0$. Since $S_G$ is
constant in time in the $C=0$ regime, the same inequality holds for the
total entropy $\mathcal{S}[\Sigma(t)]$, which proves the claim.
\end{proof}

\begin{remark}
Curvature-weighted diffusion is the sole driver of entropy growth on closed manifolds in the $C=0$ regime. The geometric contribution $S_G$ is a strict invariant of the moving-manifold kinematics, while the Boltzmann entropy increases monotonically under the inverse-curvature diffusion governed by $D^{\alpha\beta}=D_0\mathcal{B}^{\alpha\beta}$. Thus, the Second Law appears as a purely geometric consequence of the conservation Theorem~\ref{thm:CMS-balance} on the evolving manifold.
\end{remark}

\subsection{Topological flow}

Up to this point, our analysis has focused on the volume-preserving regime ($C=0$), in which the enclosed volume is fixed and, for conserved surface energy, the area is effectively preserved as well. In this setting, the manifold may still deform tangentially, but its shape evolution is kinematically constrained, and the resulting structure leads to clean conservation laws and an elegant geometric proof of the Second Law. However, this regime also reveals a limitation: geometry acts as a constraint rather than the driver of change. To uncover the full geometric contribution, we now relax the $C=0$ condition. Allowing $C\neq 0$ introduces genuine normal motion, breaks metric isometry, and activates curvature production. As we will show, this transition naturally brings topology into the entropy flow.

%Up to this subsection, our discussions have primarily centered on the volume-preserving ($C=0$) case, which has yielded remarkable theorems and provided elegant proofs of the second law. However, this approach occasionally revealed that geometry merely serves as a rule of generic underlying principles rather than setting the tone. In this context, we will relax the $C=0$ restriction and explore how it advances geometric contributions. 

\begin{theorem}[Entropy evolution]
\label{thm:geom-entropy-C}
Let $\Sigma(t)$ be a smooth evolving hypersurface with normal speed $C\neq 0$,
second fundamental form $B_{\alpha\beta}$, and geometric entropy
$S_G$ given by Corollary~\ref{cor:Sgeom}. Then, using the curvature
evolution Theorem~\ref{Theo:CE},
\begin{equation}
\frac{d}{dt} S_G[\Sigma(t)]
=
- k_G \int_{\Sigma(t)} \sigma
\bigl(\nabla^\alpha\nabla_\alpha C + C B_{\alpha\gamma}B^{\gamma\alpha}\bigr) d\Sigma
+ I[\sigma,C,B_\alpha^\alpha],
\label{eq:Sgeom-evolution-nd}
\end{equation}
where $I$ is a combination of $\sigma, \, \dot\nabla \sigma, \, C, \, B_\alpha^\alpha $. %Here
\end{theorem}

\begin{proof}
Taking into account MM fundamental Theorem~\ref{thm:CMS-transport} we compute
 \begin{align}
 &\frac{d}{dt} S_G[\Sigma(t)]=-k_G
\int_{\Sigma(t)} \dot \nabla(\sigma B_\alpha{}^{\alpha}) d\Sigma+k_G
\int_{\Sigma(t)} \sigma C (B_\alpha{}^{\alpha})^2 d\Sigma=-k_G
\int_{\Sigma(t)} (\sigma \dot\nabla B_\alpha{}^{\alpha}+\dot \nabla\sigma B_\alpha{}^{\alpha}) d\Sigma+k_G
\int_{\Sigma(t)} \sigma C (B_\alpha{}^{\alpha})^2 d\Sigma \nonumber \\
&=-k_G\int_{\Sigma(t)} \sigma\bigl(\nabla^\alpha\nabla_\alpha C + C B_{\alpha\gamma}B^{\gamma\alpha}\bigr)d\Sigma+I[\sigma,C,B_\alpha^\alpha], \quad I[\sigma,C,B_\alpha^\alpha]=k_G\int_{\Sigma(t)} B_\alpha{}^{\alpha}(\sigma C B_\alpha{}^{\alpha}-\dot\nabla\sigma )d\Sigma
 \end{align}
Note here that $I[\sigma,C,B_\alpha^\alpha]$ becomes irrelevant or low-level approximations for slowly varying $\sigma, C\sim const$ fields, and small mean curvatures $B_\alpha^\alpha \rightarrow 0$. Moreover, if $B_\alpha^\alpha \rightarrow const$ then due to conservation of $\sigma$ field, in the absence of external $\Pi$ source and $J$ flux
\[
I[\sigma,C,B_\alpha^\alpha]=k_GB_\alpha^\alpha \frac{d}{dt} \int_{\Sigma(t)} \sigma d\Sigma=0. 
\]
In that scenario, only the clean Laplacian and quadratic structures persist.  
\end{proof}

\begin{corollary}[Topology-entropy coupling]\label{topology}
Let $\Sigma(t)$ be a two-dimensional moving manifold. Then, in the quasi–static regime where
$\sigma(\xi,t)\simeq\sigma_0$ and $C(\xi,t)\simeq C_0$ are slowly varying,
the geometric entropy flow from Theorem~\ref{thm:geom-entropy-C} can be
written as
\begin{equation}
\frac{d}{dt} S_G[\Sigma(t)]
=
- k_G \int_{\Sigma(t)} \sigma\,\nabla^\alpha\nabla_\alpha C \, d\Sigma
\;+\; k\,\chi(\Sigma)
\;+\; I[\sigma,C,B_\alpha{}^{\alpha}],
\label{eq:SG-topology}
\end{equation}
where $\chi(\Sigma)$ is the Euler characteristic and
$k = 4\pi k_G \sigma_0 C_0$ is an effective constant.
\end{corollary}

\begin{proof}
Let us specialize Theorem~\ref{thm:geom-entropy-C} to $d=2$, 
%\[
%\frac{d}{dt} S_G
%=
%- k_G \int_{\Sigma(t)} \sigma
%\bigl(\nabla^\alpha\nabla_\alpha C + C B_{\alpha\gamma}B^{\gamma\alpha}\bigr) d\Sigma
%+ I[\sigma,C,B_\alpha{}^{\alpha}].
%\]
and use two dimensional the identity  $B_{\alpha\gamma}B^{\gamma\alpha} = (B_\alpha{}^{\alpha})^2 - 2K$, then we obtain
\begin{equation}\label{top}
\frac{d}{dt} S_G
=
- k_G \int_{\Sigma(t)} \sigma\,\nabla^\alpha\nabla_\alpha C\, d\Sigma
- k_G \int_{\Sigma(t)} \sigma C (B_\alpha{}^{\alpha})^2 d\Sigma
+ 2k_G \int_{\Sigma(t)} \sigma C K\, d\Sigma
+ I[\sigma,C,B_\alpha{}^{\alpha}].
\end{equation}
Next, let's absorb $- k_G \int \sigma C (B_\alpha{}^{\alpha})^2$ term into $I[\sigma,C,B_\alpha{}^{\alpha}]$ integral. Then, in the quasi–static regime $\sigma(\xi,t)\simeq\sigma_0$,
$C(\xi,t)\simeq C_0$, while due to Gauss-Bonnet theorem
\begin{equation}\label{euler}
2k_G \int_{\Sigma(t)} \sigma C K\, d\Sigma
=
2k_G \sigma_0 C_0 \int_{\Sigma(t)} K\, d\Sigma=4\pi k_G \sigma_0 C_0\,\chi(\Sigma)
= k\chi(\Sigma).
\end{equation}
where $k = 4\pi k_G \sigma_0 C_0$ and $\chi$ is the Euler characteristic, i.e, the topology of the two-surface. Plugging \eqref{euler} into  \eqref{top}  yields \eqref{eq:SG-topology}. It is important to note that we have retained only the leading contributions in the slowly varying $\sigma$ and $C$ parameters, while incorporating higher-order terms into the expression $I[\sigma,C,B_\alpha{}^{\alpha}]$. This approach demonstrates phenomenologically how topology manifests itself in the evolution of entropy.
\end{proof}

\begin{remark}
Since $\nabla_\alpha\nabla_\beta C$ is the surface Laplacian, this term induces wave-type curvature fluctuations. Combined with the topological contribution, the entropy variation couples metric fluctuations, shape fluctuations, and the topological content encoded in the Gaussian curvature. Thus, in two dimensions, geometric waves propagate directly into the topological sector of the entropy functional, producing wave-type fluctuations of topology without requiring any explicit topological change of $\Sigma(t)$.
\end{remark}

\begin{remark}[Schr\"odinger-type structure of entropy flow]\label{remark:Shrodinger}
Corollary~\ref{topology}  has direct consequence, that is a
wave–type relation of the form
\[
\partial_t s_G
\sim
-\Delta_\Sigma C + k\chi(\Sigma)+I,
\]
Since Section~\ref{def:probability} assigns probabilities through the
Boltzmann weight
$\mathbb P[\Sigma]\propto e^{\mathcal S/k_B}$, the geometric entropy is
equivalently encoded in a complex amplitude through
\(
|\psi|^2 \propto e^{S_G/k_B},
\,
S_G
\propto
\ln |\psi|^2 .
\)
Interpreting the Laplacian of the normal velocity
$\Delta_\Sigma C$ as acting on this amplitude then yields the
Schr\"odinger–type structure
\[
i\hbar \partial_t \psi
\sim
-\Delta_\Sigma \psi ,
\]
up to geometric proportionality factors.  In this sense, the Schr\"odinger
equation appears as the complexified Laplace-Beltrami dynamics already
present in the curvature–entropy evolution: the wave character of quantum
amplitudes arises as the natural counterpart of geometric entropy flow.
\end{remark}

\begin{corollary}[Global geometric entropy flow is topology in $2$D]
\label{cor:SG-topology-2D}
Let $\Sigma(t)$ be a smooth closed two-dimensional moving manifold in the quasi-static regime $\sigma(\xi,t)\simeq\sigma_0, \, C(\xi,t)\simeq C_0$ and $B_\alpha^\alpha\simeq const$. Then,
\begin{equation}
\frac{dS_G}{dt}
=
k \chi(\Sigma),
\label{dSGdt-Euler}
\end{equation}
where $S_0=4\pi k_G \sigma_0 C_0$ and $\chi(\sigma)$ is the Euler characteristics. 
\end{corollary}
\begin{proof}
It is noteworthy that Corollary~\ref{cor:SG-topology-2D} follows directly from Corollary~\ref{topology}, since for closed manifolds the Laplace--Beltrami term integrates to zero,
\[
\int_{\Sigma(t)} \nabla^\alpha\nabla_\alpha C\, d\Sigma = 0 .
\]
Consequently, $I=0$ when the conditions $\sigma(\xi,t)\simeq\sigma_0$, $C(\xi,t)\simeq C_0$, and $B_\alpha^\alpha\simeq \mathrm{const}$ are satisfied.
\end{proof}

\begin{remark}[Topological contribution to entropy production]
\label{rem:topological-entropy}
Equation \eqref{dSGdt-Euler} demonstrates that the global geometric entropy rate is directly proportional to the Euler characteristic. Consequently, the geometric entropy production is solely governed by topological properties. In regimes where the normal speed $C$ remains small, the statistical entropy $S_S$ continues to govern the usual monotonic entropy production, as in the $C=0$ sector established earlier. The geometric entropy $S_G$ therefore remains dormant under smooth deformations that preserve topology. However, when the topology of the manifold changes, the Euler characteristic changes accordingly. In such processes, the entropy production acquires a purely topological contribution proportional to the change of $\chi(\Sigma)$. This situation naturally arises in systems where topology transformations occur, such as during membrane fusion or fission events in biological membranes. In such experiments, the rate of geometric entropy production becomes directly proportional to the rate at which the genus (or more generally the topology) of the membrane changes.
\end{remark}
%-------------------------------------everything up to here is more or less chacked------------------------------------

%\section{Geometry and arising probabilities}

\section{Geometry and Fluctuation Theorem}

%\section{Geometry and Fluctuation Theorem and Arising Probabilities}

\begin{definition}[Equilibrium probability functional]\label{def:probability}
For a closed manifold $\Sigma(t)$ in thermal equilibrium, the probability
weight of a microscopic configuration is given by the Boltzmann factor
\[
\mathbb{P}[\Sigma] \;\propto\; \exp\!\left(\frac{\mathcal{S}[\Sigma]}{k_B}\right),
\]
where $\mathcal{S}[\Sigma]$ is the total entropy functional.
Thus, for two nearby configurations $\Sigma$ and $\Sigma'$,
\begin{equation}
\frac{\mathbb{P}[\Sigma']}{\mathbb{P}[\Sigma]}
=
\exp\!\left(\frac{\mathcal{S}[\Sigma']-\mathcal{S}[\Sigma]}{k_B}\right).
\label{probability}
\end{equation}
\end{definition}

\begin{lemma}[Fluctuation lemma]\label{fluctuation_lemma}
Let $\Sigma(t)$ be a smooth, closed, and moving manifold in thermal equilibrium.
Assume small curvature fluctuations $\delta B_{\alpha\beta}$ about a reference curvature tensor $\bar B_{\alpha\beta}$, while the surface density $\rho$ and energy density $\sigma$ are held fixed.
Since the geometric entropy functional is \eqref{entropy def}, then, to leading order in $\delta B_{\alpha\beta}$, the ratio of probabilities for opposite curvature fluctuations satisfies
\begin{equation}
\frac{\mathbb{P}(+\delta B_{\alpha\beta})}{\mathbb{P}(-\delta B_{\alpha\beta})}
=
\exp\!\left(
-\frac{2k_G}{k_B}
\int_{\Sigma(t)}
\sigma\,\delta B_\alpha{}^{\alpha}\, d\Sigma
\right),
\label{Geometric fluctuation}
\end{equation}
which represents a geometric fluctuation relation analogous to Crooks-type symmetry relations.
\end{lemma}

\begin{proof}
Since $\rho, \, \sigma$ are fixed, only the geometric entropy varies. From Definition~\ref{def:probability}, the entropy change under a small curvature perturbation
$B_{\alpha\beta}=\bar B_{\alpha\beta}+\delta B_{\alpha\beta}$ becomes
\[
\Delta\mathcal S
=
-\,k_G \int_{\Sigma(t)} 
\sigma\big(
B_\alpha{}^{\alpha}-\bar B_\alpha{}^{\alpha}
\big)d\Sigma =
-\,k_G \int_{\Sigma(t)}
\sigma\,\delta B_\alpha{}^{\alpha}\, d\Sigma.
\]
Fluctuations of opposite sign correspond to
$\delta B_\alpha{}^{\alpha}\rightarrow -\delta B_\alpha{}^{\alpha}$,
so the entropy difference between the two configurations becomes
\[
\Delta\mathcal S_{(+)}-\Delta\mathcal S_{(-)}
=
-2k_G
\int_{\Sigma(t)}
\sigma\,\delta B_\alpha{}^{\alpha}\, d\Sigma .
\]
Substituting this result into the Boltzmann probability ratio
\eqref{probability} yields \eqref{Geometric fluctuation}.
\end{proof}

\begin{remark}
Curvature fluctuations act as intrinsic geometric noise: their statistics are
set by the manifold's own shape rather than by any external stochastic forcing.
The Fluctuation Lemma~\ref{fluctuation_lemma} shows that entropy production is controlled by the
linear response of the geometric entropy functional to curvature deviations $\delta B_{\alpha\beta}$. On constant-mean-curvature hypersurfaces (e.g.\ a sphere), the reference 
curvature tensor $\bar B_{\alpha\beta}$ is spatially uniform and $\delta B_\alpha{}^{\alpha}=0$ so the probability ratio stays constant.  
For general shape--changing manifolds, variations in the mean curvature couple 
directly to geometric entropy production: curvature relaxation generates 
positive entropy, while spontaneous increases in curvature are exponentially 
suppressed. Thus, geometric irreversibility is encoded directly in the curvature field.
\end{remark}

\begin{remark}[Diffusion versus fluctuation entropy]
Diffusion strength and fluctuation asymmetry depend on different curvature
functionals, so there is no contradiction between the Brownian--limit remark and
the geometric fluctuation theorem. The geometric diffusion tensor involves the inverse curvature,
\(
D^{\alpha\beta} \;\propto\; \mathcal B^{\alpha\beta} = (B^{-1})^{\alpha\beta},
\)
so weakly curved (nearly flat) regions have large $\mathcal B^{\alpha\beta}$ and
therefore enhanced diffusion. In contrast, the fluctuation theorem is governed by the curvature itself through entropy variations of the form
\(
\Delta\mathcal S \propto B_\alpha{}^{\alpha}.
\)
Thus, transport properties depend on the noise covariance 
$\mathcal B^{\alpha\beta}$, while the mean curvature determines fluctuation asymmetry and entropy production. The two descriptions are dual and complementary.
\end{remark}

\begin{remark}
Lemma~\ref{fluctuation_lemma} recovers the fluctuation theorem in geometric form: at equilibrium, detailed balance appears as a curvature-fluctuation symmetry rather than invariance of
the underlying dynamics. On a static sphere with $B_{\alpha\beta}=-\tfrac{1}{R}g_{\alpha\beta}$,
the curvature is uniform, and the curvature-noise correspondence
$\mathcal B^{\alpha\beta}= -R g^{\alpha\beta}$ reduces the transport equation to the classical surface heat equation with constant diffusivity (after absorbing the sign into the effective diffusion coefficient).
\end{remark}

Subsequently, we proceed to formalize our results, ensuring that the Schrödinger equation is treated as a theorem with a magnitude, rather than a heuristic as it is in Remark~\ref{remark:Shrodinger}. To achieve this, we introduce definitions for probability densities and for complex amplitude, facilitating a seamless transition from the less rigorous Remark~\ref{remark:Shrodinger} to the fully rigorous theorem.

\begin{definition}[Local geometric probability density]\label{def:local-pG}
Let $s_G(\xi,t)$ be the geometric entropy density on $\Sigma(t)$, defined by
Corollary~\ref{cor:Sgeom}. Then, according to Definition~\ref{def:probability} the associated local geometric probability density is defined by the Boltzmann weight of $s_G$ and local complex amplitude as its complexification by exponent:
\begin{align}
p_G(\xi,t)
&:=
\frac{1}{Z_G(t)}\exp\!\left(\frac{s_G(\xi,t)}{k_B}\right),
\qquad
Z_G(t):=\int_{\Sigma(t)} \exp\!\left(\frac{s_G(\xi,t)}{k_B}\right)\,d\Sigma ,
\label{eq:pG-def} \\
\psi(\xi,t)
&:=
\sqrt{p_G(\xi,t)}\,e^{i\theta(\xi,t)}
=
\frac{1}{\sqrt{Z_G(t)}}\exp\!\left(\frac{s_G(\xi,t)}{2k_B}\right)e^{i\theta(\xi,t)}.
\label{eq:psi-def}
\end{align}
\end{definition}

\begin{lemma}[Local entropy-flow decomposition]
\label{cor:sG-flow}
In the quasi--static regime $\sigma(\xi,t)\simeq\sigma_0$,
the geometric entropy density satisfies the local flow form
\begin{equation}
\dot\nabla s_G
=
-k_G\sigma_0 \Delta_\Sigma C
+u_{\rm top}(\xi,t)+u_I(\xi,t),
\label{eq:sG-flow-one}
\end{equation}
where
\(
u_{\rm top}(\xi,t):=2k_G\sigma_0\,C(\xi,t)\,K(\xi,t), \,
u_I(\xi,t):=-k_G\sigma_0\,C(\xi,t)\,H(\xi,t)^2 .
\)
\end{lemma}

\begin{proof}
Teking invariant time derivative of $s_G$, applying curvature evalution Theorem~\ref{Theo:CE} to the mean curvature, and taking intoaccount $\dot\nabla\sigma\simeq0$ in the quasi--static regime $\sigma(\xi,t)\simeq\sigma_0, \, \sigma\simeq\sigma_0$, gives
\begin{align}
\dot\nabla s_G&
=-k_G\big((\dot\nabla\sigma) B_\alpha^\alpha+\sigma \dot\nabla B_\alpha^\alpha\big)
=-k_G\Big((\dot\nabla\sigma)\,H+\sigma\,\Delta_\Sigma C+\sigma\,C\,B_{\alpha\beta}B^{\alpha\beta}\Big) \nonumber \\
&=-k_G\sigma_0\,\Delta_\Sigma C-k_G\sigma_0\,C\,B_{\alpha\beta}B^{\alpha\beta}
=-k_G\sigma_0\,\Delta_\Sigma C-k_G\sigma_0\,C\,H^2+2k_G\sigma_0\,C\,K.
\label{eq:dotnabla-sG-start}
\end{align}
Here, we used that in two--dimensional surface, one has the identity
\(
B_{\alpha\beta}B^{\alpha\beta}=(B_\alpha{}^{\alpha})^2-2K,
\)
and $K$ is the Gauss curvature. Next, plugging designations: $u_{\rm top}(\xi,t):=2k_G\sigma_0 C(\xi,t) K(\xi,t), u_I(\xi,t):=-k_G\sigma_0 C(\xi,t) H(\xi,t)^2$ and applying it to \eqref{eq:dotnabla-sG-start} gives \eqref{eq:sG-flow-one}.
\end{proof}

\begin{theorem}[Local entropy--amplitude map and Schr\"odinger-type evolution]
\label{thm:local-schrod-one}
Let $\Sigma(t)$ be a smooth closed moving manifold evolving in the quasi--static regime
$\sigma(\xi,t)\simeq\sigma_0$ and $C(\xi,t)\simeq C_0$, such that the normal
velocity field varies slowly on the geometric time scale. Then the local
complex amplitude $\psi$ introduced in Definition~\ref{def:local-pG} obeys a
Schr\"odinger--type evolution:
\begin{equation}
i\hbar\,\dot\nabla\psi
=
-\frac{\hbar^2}{2m_{\rm eff}}\,\Delta_\Sigma\psi
+U_{\rm eff}(\xi,t)\,\psi
+R[\psi],
\label{shrodinger}
\end{equation}
where
\begin{equation}
\frac{\hbar^2}{2m_{\rm eff}}
=
\frac{i\hbar\,k_G\sigma_0}{2C_0},
\label{eq:meff}
\end{equation}
$U_{\rm eff}(\xi,t)$ is the leading multiplicative contribution of the topology
sector and the lower-order $I$-terms, and $R[\psi]$ denotes the controlled
remainder beyond the quasi--static linearization. The factor $i$ is fixed later
by the Minkowskian formulation.
\end{theorem}

\begin{proof}
In the quasi--static regime $\sigma(\xi,t)\simeq \sigma_0, \, C(\xi,t)\simeq C_0$ the geometric probability density scales with the squared normal speed, \(p_G(\xi,t)\sim C(\xi,t)^2.\) Indeed, from Definition~\ref{def:local-pG} we have $|\psi|^2=p_G$.  Since $p_G\propto e^{s_G/k_B}$ and, in the quasi--static regime, the entropy density scales with the squared normal speed, \(s_G \sim C^2 ,\) it follows that
\begin{equation}
\ln |\psi|^2
=
\frac{s_G}{k_B}
\sim
\frac{C^2}{k_B}.
\label{eq:logpsi-C}
\end{equation}
Taking into account small fluctuations $C_0\ll1$, the exponent in \eqref{eq:logpsi-C} is small and the
exponential relation for $|\psi|^2$ can be expanded to leading order,
\begin{equation}
|\psi|^2
\simeq
e^{C^2/k_B}
\simeq
1+\frac{C^2}{k_B}.
\end{equation}
Therefore, the fluctuation of the amplitude around its equilibrium value is proportional to the normal speed, \(|\psi|^2-1 \propto C^2 \). Consequently, to leading order in the quasi-static approximation, \(\psi(\xi,t)\sim C(\xi,t)\), up to normalization and phase. Hence,  $\psi(\xi,t)\sim C(\xi,t)$, taking this into account in Lemma~\ref{cor:sG-flow}, linearizing around the quasi–static background $C\simeq C_0$, and multiplying by $i\hbar/(2C_0)$ both sides of equation, the geometric entropy density rewrites as 
\begin{align}
2C_0\dot\nabla\psi
&=
-k_G\sigma_0\Delta_\Sigma\psi
+u_{\rm top}(\xi,t)+u_I(\xi,t), \label{eq:psi-pre-schrod} \\
i\hbar\,\dot\nabla\psi &=
-\frac{i\hbar k_G\sigma_0}{2C_0} \Delta_\Sigma\psi
+\frac{i\hbar}{2C_0}\bigl(u_{\rm top}(\xi,t)+u_I(\xi,t)\bigr) \label{eq:psi-ih}.
\end{align}
Finally, the lower--order contribution is decomposed into its leading
multiplicative component and a controlled remainder,
\[
\frac{i\hbar}{2C_0}\bigl(u_{\rm top}(\xi,t)+u_I(\xi,t)\bigr)
=
U_{\rm eff}(\xi,t)\,\psi
+
R[\psi].
\]
Substituting this decomposition into \eqref{eq:psi-ih} yields
the evolution equation \eqref{shrodinger}. The appearance of the factor $i$ reflects the complex structure inherent to the geometric formulation. As discussed in Definition~\ref{def:Minkowski} and Remark~\ref{rem:complex-coordinate}, the ambient Minkowski metric with signature $G_{00}=-1$ permits the temporal coordinate to be expressed as $X^0=i\tau$ (or equivalently $\xi^0=i\tau$ in intrinsic coordinates). Consequently, the velocity field $V=dX/dt$ naturally takes values in a complexified tangent space, and the constant $C_0$ may in general be complex. As a result, the amplitude $\psi$ evolves within a complex representation, which leads naturally to a Schr\"odinger--type operator governing its dynamics.
\end{proof}

\begin{remark}[Invariant operators and linearized limit]
The operators $\dot\nabla$ and $\Delta_\Sigma$ are intrinsic to the moving manifold.
Here $\dot\nabla=\partial_t - V^\alpha\nabla_\alpha$ denotes the invariant time derivative,
while $\Delta_\Sigma=\nabla_\alpha\nabla^\alpha$ is the Laplace--Beltrami operator on
$\Sigma(t)$. In the linearized regime on a weakly curved manifold patch with negligible
tangential drift ($V^\alpha\simeq 0$ and $g_{\alpha\beta}\simeq \delta_{\alpha\beta}$),
these operators reduce to the familiar flat-space expressions
\(
\dot\nabla\simeq \partial_t,
\,
\Delta_\Sigma\simeq \sum_\alpha \partial_\alpha^2 .
\)
\end{remark}

\begin{remark}[Geometric origin of the Schrödinger operator]
Theorem~\ref{thm:local-schrod-one} shows that the Schrödinger operator arises as the
leading linear approximation of the geometric entropy-flow equation on moving manifolds.
The additional curvature contributions represent geometric corrections to this dynamics.
In two dimensions, these corrections reduce to global invariants through the
Gauss--Bonnet relation, while in higher dimensions, the remaining curvature--quadratic
terms act as geometric modifications of the Schrödinger evolution.
\end{remark}

The preceding results show that the local geometric entropy-flow equation
naturally leads, in the quasi-static regime, to a Schr\"odinger-type evolution
for the entropy-induced amplitude $\psi$. In this formulation, the
Laplace--Beltrami operator arises from the geometric transport of entropy on
the moving manifold, while the remaining curvature contributions appear as
lower-order geometric corrections encoded in $U_{\rm eff}$ and $R[\psi]$.
In two dimensions, the topology sector reduces to a global invariant through
the Gauss--Bonnet relation, whereas in higher dimensions the quadratic
curvature terms provide additional geometric modifications of the dynamics.
Thus, the Schr\"odinger operator emerges as the leading linear approximation
of entropy flow in the MM framework.

\section{Geometry and Fluctuations}

In classical statistical mechanics, fluctuation theorems compare the probabilities of forward and backward trajectories generated by stochastic dynamics on a fixed phase space. Irreversibility is introduced through
external noise, representing unresolved microscopic degrees of freedom.

In the geometric formulation presented here, the situation is fundamentally different. There is no external noise: the sole fluctuating object is the geometry of the moving manifold itself. Local variations of curvature \(B_{\alpha\beta}\) and perturbations \(\delta B_{\alpha\beta}\) encode the information that, in classical theory, is attributed to stochastic forcing. Irreversibility, therefore, emerges from the intrinsic evolution of geometry, where curvature fluctuations drive the system toward or away from equilibrium,
thereby generating or suppressing geometric entropy.

The results of this section show that the MM framework provides a unified description of reversible geometry, entropy production, and topology change. Smooth curvature evolution governs dynamics within a fixed topological sector, while singular events generate discrete transitions between sectors. Together, these mechanisms establish a direct geometric link between fluctuations, irreversibility, and manifold division.

\subsection{Geometric reversibility and entropy}

The total entropy of the manifold consists of a statistical part, depending on the density field
$\rho$, and a geometric part, depending on curvature invariants.
Manifold flow is geometrically reversible when this entropy functional remains invariant along the
motion of the manifold. In the $C=0$ ambient volume conserving regime,
Corollary~\ref{thm:rev-entropy} shows that
\[
\frac{d}{dt}\mathcal S[\Sigma(t)] = 0
\]
whenever base diffusion $D_0$  vanishes. In this regime, curvature redistributions are precisely balanced: probability fluxes generated by tangential motion neither create nor destroy entropy. The geometric fluctuation Lemma~\ref{fluctuation_lemma} then describes symmetric fluctuations around this entropy level. Forward and backward curvature fluctuations are equiprobable, and the manifold exhibits geometric detailed balance: curvature deviations toward convexity or concavity occur with equal probability, producing symmetric curvature oscillations around the reversible equilibrium geometry. Once the CMS evolution leaves this reversible regime, the geometric Fokker-Planck Theorem~\ref{thm:geom-FP} generates a strictly positive entropy production rate, and the probability ratios tilt accordingly. Forward trajectories that increase geometric entropy become exponentially favored relative to their backward counterparts. The classical fluctuation theorem is recovered as a special case, but now the bias originates purely from curvature evolution rather than externally imposed stochasticity.

\begin{example}[Reversible curvature fluctuations on a sphere]
Consider a closed two–dimensional manifold $\Sigma$ given by a sphere of radius $R$,
for which the curvature tensor satisfies
\[
B_{\alpha\beta} = -\frac{1}{R} g_{\alpha\beta},
\qquad
B_\alpha{}^{\alpha} = -\frac{2}{R}.
\]
In the reversible regime $C=0$ and $D_0=0$, Corollary~\ref{thm:rev-entropy} implies that the entropy is conserved.
%\[
%\frac{d}{dt}\mathcal S[\Sigma(t)] = 0.
%\]
Hence, the geometric entropy remains constant during the motion and small perturbations of the curvature,
\(
B_{\alpha\beta} \rightarrow B_{\alpha\beta}+\delta B_{\alpha\beta},
\)
produce entropy variations determined by Lemma~\ref{fluctuation_lemma}. 
Since the reference curvature of the sphere is spatially uniform,
positive and negative perturbations $\pm \delta B_{\alpha\beta}$ lead to
equal and opposite entropy changes. Consequently,
\[
\frac{\mathbb P(+\delta B_{\alpha\beta})}{\mathbb P(-\delta B_{\alpha\beta})}=1 .
\]
Thus, curvature fluctuations around the spherical equilibrium geometry
are statistically symmetric: convex and concave deviations occur with equal
probability. This provides a concrete realization of geometric detailed balance
in the reversible regime of the moving-manifold dynamics.
\end{example}
In the nonlinear regime, curvature fluctuations are not necessarily confined to small magnitudes. Convex regions may undergo a progressive increase in convexity, while concave regions may deepen, resulting in the concentration of curvature within narrow necks. In the extreme case, these necks may collapse, leading to the emergence of singular points at which the manifold bifurcates into two distinct daughter geometries. These geometric bifurcations are associated with alterations in the Euler characteristic, thereby activating the topological contribution to entropy production as elucidated in the subsequent discussion.
\begin{example}[Curvature waves on a cylindrical manifold]
Consider a cylindrical surface $\Sigma$ of radius $R$ and axial coordinate $z$.
The principal curvatures are
\[
\kappa_1 = -\frac{1}{R}, \qquad \kappa_2 = 0,
\]
so that the mean curvature is $B_\alpha{}^{\alpha} = -\frac{1}{R}$. In the reversible regime $C=0$ and $D_0=0$, the total entropy remains constant, but curvature perturbations may propagate along the axis of the cylinder. Small fluctuations of the form
\(
\delta B_{\alpha\beta}(z,t)
\)
travel along the surface without changing the global entropy, producing
wave-like redistributions of curvature along the cylindrical direction. Such traveling curvature modes resemble the peristaltic geometric waves observed in many cylindrical biological membranes, including nerve axons, where membrane deformations propagate while the global entropy balance
remains approximately reversible.
\end{example}
The cylindrical example illustrates an important distinction between normal and tangential manifold motion. In the present framework, wave-like behavior does not require a normal deformation of the surface. Even in the reversible regime $C=0$, curvature-velocity coupling through the invariant quantity $V^\alpha V^\beta B_{\alpha\beta}$ allows curvature perturbations to evolve dynamically along the manifold. On a cylindrical surface, where one principal curvature vanishes and the other is constant, tangential motion along the axial direction naturally supports oscillatory redistributions of curvature. Because the entropy remains conserved in this regime, the fluctuation symmetry described by Lemma~\ref{fluctuation_lemma} enforces detailed balance between positive and negative curvature perturbations. The resulting dynamics therefore resemble standing curvature waves along the cylinder: locally convex regions may temporarily flatten or even become weakly concave before returning to equilibrium. When the system leaves the reversible regime, entropy production breaks this symmetry, and the curvature oscillations acquire a directional bias, producing traveling geometric waves along the cylindrical surface.

\subsection{Topology and irreversible jumps}

The geometric entropy functional introduced earlier contains an
explicit topological sector in two dimensions. As shown in
Corollary~\ref{topology} (see also \eqref{eq:SG-topology}), the entropy
flow separates into a Laplace--Beltrami contribution and a term
proportional to the Euler characteristic $\chi(\Sigma)$ through the
Gauss--Bonnet relation.  As long as the topology of the manifold is
preserved, this contribution is constant and the CMS dynamics remain
confined to a single topological sector.

The local structure of this entropy flow is given by the entropy
evolution law derived in Lemma~\ref{cor:sG-flow} (equation
\eqref{eq:sG-flow-one}).  In that expression the Laplace--Beltrami
operator controls the redistribution of curvature across the surface,
while the remaining terms couple the normal motion $C$ to the Gaussian
curvature and the mean--curvature sector.  Consequently, when the
manifold evolves in the volume--preserving regime $C=0$, curvature
fluctuations propagate along the surface but the topological sector
remains dynamically frozen.

Once $C\neq0$, however small, this situation changes qualitatively.
Normal motion activates the curvature contributions that contain the
Gaussian curvature $K$.  Through the Gauss--Bonnet identity these terms
are directly tied to the Euler characteristic and therefore to the
topological content of the surface.  In this sense, normal motion makes
the topology of the manifold dynamically visible to the entropy flow.

Through the entropy--probability correspondence introduced in
Definition~\ref{def:local-pG}, the same geometric entropy density
governs the probability weight of microscopic configurations and the
associated complex amplitude.  As shown in
Theorem~\ref{thm:local-schrod-one} (see also equation \eqref{shrodinger}
for the resulting evolution law), the Laplace--Beltrami part of the
entropy flow induces the Schr\"odinger--type dynamics of the amplitude,
while the curvature contributions act as geometric corrections to that
evolution.

A genuine topology change, such as a pinch, neck collapse, or handle
creation, cannot occur under smooth CMS evolution.  Such events
correspond to singular moments at which the smooth manifold
description breaks down and the Euler characteristic changes
discretely.  At these singular transitions the Gauss--Bonnet sector of
the entropy functional undergoes a finite jump.  The CMS equations
therefore describe reversible dynamics within each fixed topological
sector, whereas transitions between inequivalent topologies appear as
intrinsically irreversible entropy jumps.

Note that in the limit of a large number of handles or holes, the Euler
characteristic may be viewed as a topological density rather than a strictly
discrete invariant. In that regime it is convenient to work directly with the
Gauss--Bonnet representation
\(
\int_{\Sigma} K\,d\Sigma
\)
instead of the discrete quantity $\chi(\Sigma)$. In that regime, it is convenient to represent the topological sector through the Gauss--Bonnet integral $\int_{\Sigma} K\,d\Sigma$, while keeping in mind that this quantity remains constant along each smooth branch of the evolution. The effective continuous behavior arises from the accumulation of many discrete topology-changing events rather than from a smooth variation of the Gauss--Bonnet integral itself.

 Although the detailed analysis of this continuum limit is not the main goal of the present work, it fits naturally within the MM calculus and can in principle be treated rigorously within the same framework.

In this limit, successive topology--changing events produce a sequence of small
entropy increments that effectively approximate a continuous topological entropy
flow. Situations of this kind arise, for example, in systems with densely
connected structures where a large number of holes or handles appear. One
illustration is provided by certain classes of topological materials whose
electronic or geometric structures are characterized by large numbers of
topological defects or loops. Another example appears in the topological phase
picture, where particles are modeled as finite spheres in chemical contact. Such spheres may either arrange themselves into a single non--self--intersecting chain or form configurations with many holes and handles, thereby defining distinct topological phases of the system. In the high-connectivity limit, these configurations naturally approach a continuum description in which the discrete Euler characteristic is replaced by a smooth curvature measure.

\begin{example}[Topology change in membrane fusion and fission]
A well-known physical illustration of topology-driven entropy change is
provided by biological membranes. In many membrane fusion and fission
experiments, closed lipid vesicles that initially possess spherical
topology ($\chi=2$) can undergo morphological transitions in which the
membrane develops a handle and temporarily assumes a toroidal
configuration ($\chi=0$). Such transformations are commonly observed in
processes such as vesicle fusion, pore formation, and membrane
remodeling.

These topological transitions are experimentally known to be accompanied
by finite jumps in the membrane free energy, reflecting the energetic
cost of creating or removing handles or pores. Within the present
framework this behavior is naturally interpreted through the geometric
entropy law
\[\frac{dS_G}{dt}=k\,\chi(\Sigma),\]
which predicts that transitions between different topological sectors must involve discrete entropy
changes. The spherical and toroidal configurations, therefore, correspond to distinct entropy levels, and the transition between them necessarily occurs through a singular moment of neck formation or collapse where the smooth manifold description breaks down.

In this sense, experimentally observed membrane fusion or fission events
provide a natural realization of the topology--entropy coupling implied
by the MM calculus: topology-changing events manifest themselves as
finite entropy jumps, which in thermodynamic terms correspond to the
observed discontinuities in free energy \cite{Mayer2002}.
\end{example}

\subsection{Curvature-noise mismatch and division}

According to the Curvature Configuration Theorem~\ref{Ostwald theorem}, the curvature--kinetic invariant
\(
V^\alpha V^\beta B_{\alpha\beta}
\)
enters directly into the exponential weight determining the equilibrium configuration distribution \eqref{eq:geom-ostwald}. Thus, curvature fluctuations do not act externally but are encoded intrinsically in the geometric weighting of configurations. In the moderate-noise regime, this coupling remains self-regulating. The curvature--kinetic term entering the exponential weight generates a feedback mechanism: deviations of
\(
V^\alpha V^\beta B_{\alpha\beta}
\)
modify the local configuration weight and drive the system toward configurations that remain close to the reversible entropy manifold. In this regime, geometric fluctuations remain balanced, and the system
stays near geometric detailed balance.

When the geometric diffusion scale becomes sufficiently large, this feedback mechanism ceases to function. The same curvature-kinetic invariant that previously stabilized fluctuations now amplifies them. Curvature deviations reinforce geometric diffusion rather than opposing it, and the entropy production term derived in Theorem~\ref{thm:entropy-irreversible} dominates the evolution. Consequently, curvature extrema become sharper rather than relaxed, and the configuration weight develops multiple locally preferred states. In this regime, the manifold no longer favors global relaxation but instead evolves toward geometric separation. Regions of distinct curvature emerge and persist, corresponding to locally preferred configurations of the curvature-kinetic action. In the strong mismatch limit, this process leads to geometric division events, in which the manifold splits into distinct domains that subsequently evolve as
separate geometric configurations.

\begin{remark}
The instability criterion follows directly from the entropy production
structure of Section~4. In the reversible regime ($D_0=0$ and $C=0$),
both statistical and geometric entropy remain invariant. When diffusion
is present ($D_0\neq0$), the entropy production rate is given by
\[
k_B D_0 \int_{\Sigma(t)}
\mathcal{B}^{\alpha\beta}
\frac{\nabla_\alpha\rho\,\nabla_\beta\rho}{\rho}\, d\Sigma,
\]
which is strictly nonnegative. Whether this contribution suppresses or
amplifies fluctuations is determined by its interplay with the
curvature--kinetic invariant
\(
m\,V^\alpha V^\beta B_{\alpha\beta}
\)
entering the configuration weight
\eqref{eq:geom-ostwald}. Thus, the curvature--noise mismatch is not an
additional assumption but follows directly from the entropy production
law and the geometric weighting of configurations.
\end{remark}

\begin{remark}[Bridge to geometric quantum weights]
The same quadratic form
\(
B_{\alpha\beta}V^\alpha V^\beta
\)
governs both entropy production and configuration weighting. In the
Onsager--Machlup functional \eqref{eq:OM-Eulerian}, it appears as the
geometric rate functional controlling transition probabilities. In the
path--integral formulation of Section~3, this identical structure defines
the action whose exponential weight generates configuration amplitudes.
Thus, the mechanism responsible for geometric division in the
high--noise regime is mathematically the same object that determines the
weight of paths in the geometric amplitude representation, linking
irreversible entropy production and geometric transition weights within
a single curvature--kinetic structure.
\end{remark}

\begin{example}[Continuous limit of topology-driven entropy flow]
In two dimensions, Corollary~\ref{topology} shows that the geometric entropy
flow decomposes into a continuous geometric contribution and a discrete
topological sector. In the special case $C=0$, the Laplace--Beltrami term
vanishes and the entropy flow reduces to the purely topological law described
in Corollary~\ref{topology}, so that the evolution is governed entirely by the
Euler characteristic. Since $\chi(\Sigma)$ is a discrete invariant,
topology-changing events (such as handle creation or neck collapse) produce
finite jumps in the entropy flow, and the evolution proceeds through
transitions between distinct topological sectors.

In the general quasi--static regime $C\neq0$, however, the full entropy
evolution law of Corollary~\ref{topology} remains active. Although the
Gauss--Bonnet relation ensures that the total curvature integral is constant
along each smooth branch of the evolution, the Laplace--Beltrami contribution
generates continuous curvature dynamics. This term drives the redistribution
and amplification of curvature and is responsible for the nonvanishing entropy
variation within each fixed topological sector.

When curvature fluctuations become sufficiently strong, this Laplacian-driven
evolution produces localized curvature concentration, leading to the formation
of necks and singular configurations. At such singular events the topology of
the manifold changes, and the Euler characteristic undergoes a discrete jump,
resulting in a finite change of the entropy flow.

In regimes where many such topology-changing events occur, these discrete
transitions accumulate and effectively approximate a continuous topological
dynamics. In this sense, curvature evolution induces not only metric
deformations but also an effective evolution of the topological content of the
manifold. Geometrically, this provides a natural mechanism for growth and
division: curvature amplification leads to neck formation and subsequent
separation into distinct geometric domains, linking entropy production,
curvature dynamics, and topology within a unified MM framework.
\end{example}

In higher dimensions, the structure of the entropy evolution simplifies
in an essential way. While in two dimensions the Gaussian curvature
separates into a distinct topological sector through the Gauss--Bonnet
relation, the general entropy evolution law of
Theorem~\ref{thm:geom-entropy-C} shows that all contributions are encoded
in curvature invariants of the form $B_{\alpha\gamma}B^{\gamma\alpha}$.
In this setting, topology no longer appears as an independent term but is
implicitly contained within the full curvature structure of the evolving
manifold.

Consequently, the distinction between geometric evolution and topological
transitions become less explicit: curvature amplification, entropy
production, and manifold reconfiguration are governed by a unified
curvature dynamics. In particular, the mechanism of curvature-driven
instability and division described above extend naturally to higher
dimensions without the need to isolate a separate topological sector.

\section{Curvature configurations and quantum amplitudes}

According to the Curvature Configurations Theorem~\ref{Ostwald theorem},
the generalized Ostwald relation \eqref{eq:geom-ostwald} assigns a weight
to each instantaneous geometric state of the evolving manifold $\Sigma(t)$.
This weight is generated by the curvature--kinetic invariant
\(B_{\alpha\beta}V^\alpha V^\beta\).
As established in Definition~\ref{om} and Corollary~\ref{cor:OM}, the same
invariant defines the Eulerian Onsager-Machlup functional
\eqref{eq:OM-Eulerian} and the corresponding short-time transition weight
\eqref{OM-short}. Along a deformation path \(u^\alpha(t)\) describing successive geometric
states of the evolving manifold, with \(V^\alpha=\dot u^\alpha\), the
Ostwald weight admits a pathwise interpretation. It is generated by the
same geometric quadratic form in the path velocities that governs the
Onsager--Machlup kernel. In this sense, curvature configurations,
traditionally expressed through Ostwald-type concentration relations,
become statistical weights assigned directly to the geometric states of
the evolving manifold.

By Corollary~\ref{qm theorem}, the same structure extends from
infinitesimal geometric deformations to the full time-accumulated path
integral \eqref{eq:geom-path-prob}. Consequently, the stochastic
propagator and the global path-integral weight arise from the same
curvature--kinetic action. The Onsager--Machlup regime describes the
local (infinitesimal) realization of this action, while the path-integral
representation describes its global accumulation along the deformation
history of the manifold.
The connection to quantum amplitudes follows from the Minkowskian nature
of the ambient space, as described in Remark~\ref{Minkowski}. When the
ambient geometry is Lorentzian, the same curvature--kinetic action that
generates the dissipative weight \eqref{eq:geom-path-prob} acquires an
oscillatory phase. The geometric transition amplitude is therefore
obtained from the same functional by passing from the thermal/stochastic
realization to the Minkowskian realization. The distinction between
Ostwald-type weights and Feynman-type amplitudes is thus not geometric,
but arises from the signature through which the same geometric action is
realized.

Accordingly, both stochastic weights and quantum amplitudes are generated
by the same curvature--velocity quadratic structure associated with the
evolving manifold. The stochastic description corresponds to a real
exponential weight, while the quantum description corresponds to a
complex oscillatory phase. Thermal, stochastic, and quantum regimes
therefore represent different realizations of a single geometric action.

\begin{remark}
Throughout this section, a ``configuration'' refers to an instantaneous
geometric state of the evolving manifold $\Sigma(t)$, rather than to a
point in an abstract configuration manifold in the sense of classical
statistical mechanics. All weights and transition amplitudes are thus
assigned directly to geometric states and their deformation histories.
\end{remark}

\begin{remark}
The Curvature Configurations Theorem~\ref{Ostwald theorem},
Definition~\ref{om}, Corollary~\ref{cor:OM}, and the path-integration
Corollary~\ref{qm theorem} together show that the same curvature--kinetic
functional governs equilibrium weights, Onsager--Machlup transition
probabilities, and the geometric path integral. The stochastic and
quantum descriptions therefore arise from different realizations of the
same curvature--velocity action rather than from distinct mechanisms.
\end{remark}

\begin{remark}
Imaginary time is not introduced as an auxiliary assumption. As discussed
in Remarks~\ref{rem:complex-coordinate} and \ref{Minkowski}, it follows
from the Lorentzian signature of the ambient Minkowski space itself. The
Euclidean and Minkowskian forms of the path weight are therefore two
signature realizations of the same geometric functional.
\end{remark}

\begin{example}[Curvature-induced thermal--quantum crossover]
The geometric action appearing in the Curvature Configurations Theorem~\ref{Ostwald theorem}, the Onsager--Machlup functional \eqref{eq:OM-Eulerian}, and the path-integral representation \eqref{eq:geom-path-prob} admits two complementary realizations: a real exponential weight in the thermal/stochastic sector and an oscillatory phase in the Minkowskian sector.The crossover between these regimes is obtained by evaluating the same curvature-kinetic action on a characteristic geometric scale. In a locally isotropic regime,
\[
B_{\alpha\beta}\sim \frac{1}{R} g_{\alpha\beta},
\]
where $R$ is the characteristic curvature radius. As discussed in Remark~\ref{rem:complex-coordinate}, the corresponding geometric time scale is $\tau_R \sim R/v$, where $v$ is the characteristic propagation
speed. Comparing the thermal weights \eqref{eq:OM-Eulerian}, \eqref{OM-short}
and the quantum phase \eqref{eq:geom-path-prob} for the same geometric
action yields
\begin{equation}
\frac{1}{k_B T}\sim \frac{\tau_R}{\hbar}\sim \frac{R}{\hbar v},
\qquad
k_B T\sim \frac{\hbar v}{R}.
\label{cross}
\end{equation}
Thus, the transition between thermal and quantum regimes is governed by
the inverse curvature scale of the manifold. Physical effects such as the
Unruh and Hawking temperatures, as well as confinement-induced quantum
gaps, arise as particular realizations of this geometric scaling.
\end{example}

\section{Thermal-Quantum crossovers}

The scaling relation \eqref{cross}, derived in the preceding example, elucidates that the inverse curvature scale of the evolving manifold governs the transition between thermal and quantum regimes. Specifically, the crossover condition is determined by the characteristic curvature radius $R$ and the corresponding geometric time scale $\tau_R$. For a medium confined by a moving manifold, the same curvature-kinetic invariant underlying \eqref{eq:geom-ostwald} and \eqref{eq:geom-path-prob} naturally extends to the bulk geometric action. In this context, the energy and the associated Boltzmann weight originate from the volumetric curvature-velocity coupling, consistent with the surface formulation developed in Sections~3 and~7. Consequently, both thermal weights and quantum amplitudes originate from the same curvature-velocity quadratic structure, evaluated at different geometric scales. Therefore, the distinction between thermal and quantum regimes is not due to different geometric objects, but rather the scale at which the same geometric action becomes dominant. This clarifies the geometric origin for a range of physical phenomena. Effects such as Unruh and Hawking temperatures, as well as confinement-induced quantum gaps, discussed in detail below, correspond to particular realizations of the curvature scaling \eqref{cross}, where the relevant length scale is set by acceleration, the horizon radius, or geometric confinement. 

The relation \eqref{cross} reflects a deeper geometric structure. Both the thermal weight and the quantum phase are generated by the same curvature-kinetic invariant \(B_{\alpha\beta}V^\alpha V^\beta\), and thus involve the same geometric tensor. The distinction between the two does not arise from different geometric objects, but from the scale at which this invariant is realized. In the locally isotropic regime \(B_{\alpha\beta}\sim H g_{\alpha\beta}\), with \(H\sim 1/R\), the curvature tensor reduces to a scalar curvature scale. The thermal and quantum coefficients then act on the same geometric object evaluated at different characteristic scales. The crossover condition \eqref{cross} therefore expresses the fact that the same curvature-kinetic action, when evaluated at the scale \(R\), produces either thermal or quantum behavior depending on the relative magnitude of \(k_B T\) and \(\hbar\). An additional distinction arises in the temporal structure of the two realizations. The thermal weight \eqref{OM-short} is generated locally in
time: it is constructed incrementally through infinitesimal steps and therefore reflects a time-local (Markovian) realization of the curvature-kinetic action. By contrast, the quantum amplitude \eqref{eq:geom-path-prob} is defined globally over entire deformation histories. The oscillatory phase depends on the full time-integrated action and is therefore intrinsically nonlocal in time. It should be emphasized that the time parameter appearing in the above (see Definition~\ref{def:Minkowski}, Remark~\ref{rem:complex-coordinate}) expressions is not proper time, but the parametric time associated with the evolution of the manifold. The deformation path \(u^\alpha(t)\) tracks successive geometric states of the manifold, and the action is defined with respect to this intrinsic evolution parameter rather than
a physical time coordinate. This distinction is essential, as the geometric action depends only on the structure of the deformation history, not on a specific temporal parametrization of spacetime. Thus, the same curvature--kinetic functional admits two complementary realizations: a time-local, dissipative form in the thermal regime and a time-global, phase-coherent form in the quantum regime.

The scaling relation \eqref{cross} reveals that several well-established quantum-geometric scales emerge directly from the same curvature--motion principle. When $c$ is the speed of light, quantum effects become relevant when the geometric action associated with a characteristic curvature radius or confinement length becomes comparable to the thermal scale.

\begin{example}[Unruh temperature]
A uniformly accelerated observer with proper acceleration $a$ traces a
worldline whose curvature radius is \(R = c^2/a.\) Substituting this into \eqref{cross} yields
\(k_B T \sim \hbar c/R = \hbar a/c.\)
Up to the universal factor $2\pi$, this reproduces the Unruh temperature \cite{Unruh1976}
\[
k_B T_{\mathrm U} = \frac{\hbar a}{2\pi c}.
\]
Geometrically, the effect arises because strong acceleration reduces the curvature radius of the observer's trajectory, bringing the geometric action scale into correspondence with the thermal scale. The observed temperature is therefore a direct manifestation of curvature in the observer's worldline.
\end{example}

\begin{example}[Hawking temperature]
For a black hole, the Schwarzschild radius $R_s$ sets the curvature scale
of spacetime at the horizon. Applying \eqref{cross} with $R=R_s$ gives
\(
k_B T \sim \hbar c/R_s.
\)
Up to a numerical factor, this reproduces the Hawking temperature \cite{Hawking1975}
\[
k_B T_{\mathrm H} = \frac{\hbar c}{4\pi R_s} .
\]
Hence, Hawking radiation arises due to the curvature scale of spacetime near the horizon, driving the geometric action to transition into the quantum regime. Vacuum fluctuations acquire a thermal character when the curvature radius becomes comparable to the intrinsic quantum scale.
\end{example}

\begin{example}[Casimir effect]
In the Casimir effect, two boundaries separated by a distance $L$ impose geometric constraints on admissible modes. The confinement length $L$ acts as an effective curvature scale, and \eqref{cross} yields
\[
E \sim \frac{\hbar c}{L}.
\]
This recovers the characteristic scaling of the Casimir energy \cite{Casimir1948}. The effect arises because geometric confinement restricts the manifold deformation modes, forcing the curvature-kinetic
action to be realized at the scale $L$.
\end{example}

\begin{example}[Finite-size quantum gaps]
In finite systems, such as confined quantum media or discrete structures, the characteristic length $L$ determines the spacing of energy levels. Applying \eqref{cross} yields
\[
\Delta E \sim \frac{\hbar c}{L},
\]
which is consistent with standard finite-size scaling relations \cite{Fisher1998}. Here, the discrete spectrum arises from the restriction of admissible geometric states of the system. The curvature-kinetic action evaluated on the finite domain enforces quantization through geometric confinement.
\end{example}

These examples illustrate that a diverse array of phenomena traditionally perceived as distinct share a common geometric origin. In each instance, the pertinent observable emerges when the characteristic curvature or confinement scale of the manifold becomes comparable to the quantum scale defined by $\hbar$. More generally, the crossover condition \eqref{cross} implies that quantum behavior dominates when the characteristic curvature radius satisfies
\[
R \sim \lambda_T, \qquad \lambda_T \sim \frac{\hbar c}{k_B T},
\]
where $\lambda_T$ is the thermal de Broglie length. In this regime, the curvature-kinetic action per unit scale matches the thermal scale, and the oscillatory (Minkowskian) realization of the geometric action becomes dominant. Conversely, when the curvature scale is large, $R \gg \lambda_T$, the geometric action is effectively realized in its real (Boltzmann) form. Fluctuations are then governed by dissipative dynamics, resulting in classical behavior. Between these limits lies the geometric crossover region defined by \eqref{cross}, where thermal and quantum weights become comparable. In this regime, the same curvature-kinetic action simultaneously supports stochastic diffusion and wave-like propagation, producing a mixed stochastic-quantum behavior governed entirely by the geometric scale of the evolving manifold.

Consequently, we establish that the same geometric scaling \eqref{cross} governs the realization of the curvature-kinetic action in the moving-manifold framework. The characteristic curvature radius $R$ or confinement scale $L$ determines whether the action is realized in its thermal or quantum form, while $k_B$ and $\hbar$ act as conversion factors that relate the same geometric quantity to probabilistic weights in the respective regimes. When the curvature evolution satisfies the reversible balance laws, the associated geometric entropy remains constant, and forward and backward deformation histories are statistically symmetric. In contrast, when the evolution produces entropy, this symmetry is broken, and forward histories become exponentially favored, yielding the fluctuation theorem in its geometric form. Topological changes introduce an additional level of irreversibility: when the manifold itself changes its domain or topology, the entropy exhibits discrete jumps associated with the change of the geometric state space.

Therefore, the moving-manifold framework provides a unified geometric origin for stochastic irreversibility and for the weighting structures of thermal and quantum theory. Stochastic behavior, thermodynamic irreversibility, and quantum path amplitudes all arise from the same curvature-velocity action, with their differences determined by the scale and realization of this underlying geometric structure.

The examples discussed herein serve as representative rather than exhaustive. The same geometric scaling is anticipated to manifest in a wide range of systems where curvature or confinement dictates the predominant length scale. Our objective in this section is to elucidate that the distinction between classical and quantum behavior is determined by the extent of geometric constraint: when curvature exerts a stringent restriction on permissible deformation histories, the dynamics exhibits classical characteristics, whereas in weakly constrained regimes, a diverse array of histories contributes, resulting in quantum-like behavior. A comprehensive investigation of such correlations is deferred for future research.

\section{Discussion}

The central structural conclusion of this work is that stochastic behaviour in continuum systems originates from curvature fluctuations and, in singular events, from topological transitions of the evolving manifold. The geometry of the state space itself generates stochastic features, rather than stochasticity arising from externally imposed random perturbations. Irreversibility appears as geometric entropy production (Theorem~\ref{thm:entropy-irreversible}), while discrete entropy jumps arise from transitions between distinct topological sectors. Classical stochastic objects such as diffusion tensors, Onsager--Machlup functionals (Corollary~\ref{cor:OM}), Fokker--Planck currents (Theorem~\ref{thm:geom-FP}), fluctuation theorems, and path-integral weights (Corollary~\ref{qm theorem}) therefore emerge as geometric expressions attached to manifold motion rather than as independent probabilistic postulates. Stochastic calculus is thus intrinsic to the MM framework.

The geometric origin of diffusion follows directly from the curvature--noise correspondence developed in Section~3 and is governed by the geometric Fokker--Planck equation (Theorem~\ref{thm:geom-FP}). Tangential motion along the manifold generates an effective diffusion process whose local intensity is controlled by curvature: regions of large curvature suppress transport, while nearly flat regions enhance it. Drift, diffusion, and entropy production are therefore governed by the extrinsic geometry of the evolving manifold. In higher-dimensional surfaces, spatial heterogeneity of \(B_{\alpha\beta}\) produces curvature wells and ridges that act respectively as trapping regions and transport channels. Amplification of curvature variations leads to geometric bifurcations and domain separation. These effects arise intrinsically from the geometric evolution laws without introducing external stochastic forcing.

Throughout this work, emphasis was placed on the \(C=0\) regime, where tangential motion dominates and no volumetric change occurs. In this incompressible, curvature--dominated limit, the MM effective pressure reduces to
\[
P = \rho\,V^\alpha V^\beta B_{\alpha\beta},
\]
(see Corollary~\ref{cor:CMSmomentum}, derived from Theorem~\ref{thm:CMS-EOM}), and the geometric stochasticity principle emerges in its simplest form. However, the full moving--manifold dynamics allows nonzero normal speed \(C\). In the general compressible regime, additional coupling terms involving \(\dot{\nabla}C\), \(V^\beta\nabla_\beta C\), and curvature--normal interactions modify the entropy balance and pressure evolution. The curvature invariant \(B_{\alpha\beta}B^{\alpha\beta}\) then evolves through both tangential and normal contributions, and topology becomes dynamically coupled to motion through the full geometric balance laws. The stochastic theory developed here should therefore be viewed as the leading tangential limit of the complete MM dynamics. When normal motion and external fluxes are significant, the full pressure balance law (Theorem~\ref{thm:CMS-EOM}) must be employed rather than the reduced \(C=0\) form.

For completeness, we note that the general theory of moving manifolds admits a full time-integrated pressure balance law, obtained directly by integrating the normal equation of motion in Theorem~\ref{thm:CMS-EOM}. In this form, the pressure incorporates the coupled contributions of tangential motion, normal velocity, curvature, and ambient fluxes, and may be written as
\[
P = \int_{t_0}^{t_1} \left( \partial_A \left[ V^A \left( \rho (\dot{\nabla} C + 2 V^\beta \nabla_\beta C + V^\beta V^\gamma B_{\beta\gamma}) - P + \sigma B_\alpha{}^\alpha \right) \right] - \partial_A \partial_t F^A \right) dt.
\]
This expression makes explicit that the \(C=0\) regime considered throughout this work is a deliberate geometric reduction rather than a limitation of the formalism. Accordingly, the stochastic theory developed here should be viewed as the leading tangential limit of the complete MM dynamics. When normal motion and external fluxes are significant, the full pressure balance law must be employed. Notably, partial extension to the \(C \neq 0\) regime already leads to Schrödinger-type evolution (Theorem~\ref{thm:local-schrod-one}), demonstrating that quantum behaviour emerges within the same geometric framework. Taken together, these results show that stochastic behaviour, thermodynamic irreversibility, and quantum transition amplitudes arise from the same curvature--velocity quadratic structure. The moving manifold framework thus provides a unified geometric origin for diffusion, entropy production, fluctuation asymmetry, and path-integral weighting. These phenomena are interpreted differently in thermal, classical, and quantum regimes, but are generated by the same underlying geometric structure of manifold motion.

The quantum limit is not introduced as an independent postulate but follows from the local entropy--flow decomposition and the Schr\"odinger-type evolution (Theorem~\ref{thm:local-schrod-one}), together with the developments of Section~4, where the geometric entropy density is constructed and shown to define a complex amplitude whose temporal evolution is governed by the same curved Laplace--Beltrami operator. When embedded in a Minkowskian ambient space, the curvature--velocity quadratic form acquires a Lorentzian signature, and the entropy evolution transforms into an oscillatory phase evolution. Consequently, the amplitude obeys a Schr\"odinger-type equation up to controlled geometric prefactors, rendering the standard Schr\"odinger equation an approximation rather than an exact identity within the full MM framework.

The same geometric structure simultaneously generates fundamental laws across different regimes: the second law of thermodynamics (Theorem~\ref{thm:entropy-irreversible}), the geometric Fokker--Planck equation (Theorem~\ref{thm:geom-FP}), the Onsager--Machlup transition law (Corollary~\ref{cor:OM}), the path-integral representation (Corollary~\ref{qm theorem}), and the Schr\"odinger-type evolution (Theorem~\ref{thm:local-schrod-one}). Their emergence from a single curvature--velocity quadratic form shows that they are not independent principles but different realizations of the same underlying geometric flow of moving manifolds.

From the perspective of existing literature, stochastic processes on
manifolds, Onsager-Machlup functionals, and path-integral formulations
have traditionally been introduced as probabilistic or variational
structures defined on a fixed geometric background \cite{Seifert2012Review}. In contrast, the
present framework derives these objects directly from the intrinsic
evolution of the manifold itself. In this sense, diffusion, stochastic
weights, and quantum amplitudes do not arise as independent constructs,
but emerge from a single curvature-driven geometric mechanism. The
results obtained here therefore provide a structural unification of
concepts that are typically treated separately in stochastic analysis,
statistical physics, and quantum theory.

The appearance of the imaginary unit does not rely on analytic continuation but follows from the Lorentzian character of the ambient geometry (see Remark~\ref{Minkowski}). The same curvature--velocity invariant that generates dissipative weights in the thermal regime, therefore generates quantum phases in the relativistic regime. Classical diffusion, thermodynamic irreversibility, and quantum evolution are thus different manifestations of a single geometric operator acting on the evolving manifold. In this sense, the MM framework provides a structural rather than heuristic unification of classical and quantum descriptions.

\section{Conclusion}

Geometry transcends mere accommodation of stochasticity; probabilistic behaviour arises as a consequence of deterministic geometric evolution.  In the moving manifold framework, curvature fluctuations and topological transitions of the evolving manifold generate the phenomena ordinarily attributed to external noise. The curvature--kinetic invariant
\(
\rho B_{\alpha\beta}V^\alpha V^\beta
\)
governs short-time Onsager--Machlup fluctuations, long-time Ostwald-type thermodynamic weights, and the geometric action underlying path-integral amplitudes. Stochastic behaviour is therefore a coarse-grained manifestation of deterministic geometric evolution rather than an independent probabilistic postulate. Likewise, the Second Law ceases to be an axiom: entropy monotonicity follows directly from curvature-controlled dynamics through the positive-definite geometric diffusion structure.

This framework establishes a direct bridge between classical and quantum dynamics. The same geometric functional governs both regimes: classical behaviour corresponds to its time-local, curvature-driven realization, while quantum behaviour corresponds to its time-global, phase-coherent realization. Classical mechanics and quantum mechanics thus arise not as distinct dynamical principles, but as different signatures of the same curvature--velocity action. The path-integral formulation emerges directly from the geometry of moving manifolds, without introducing stochastic forcing or quantum fluctuations as independent inputs. Within this interpretation, the wavefunction represents the curvature-encoded state of the evolving manifold, while operators generate geometric flow. Quantum amplitudes become weights of the curvature--kinetic action. The distinction between stochastic and quantum regimes is therefore not structural but representational: real weights in the thermal sector and complex phase weights in the Minkowskian sector, with the underlying geometric functional remaining unchanged.

The moving manifold framework thus provides a deterministic geometric origin for fluctuation, irreversibility, and quantum superposition. The curvature of the evolving manifold, coupled to its motion, acts as the unifying mechanism from which these phenomena emerge as different realizations of the same geometric structure. The theory further yields concrete geometric consequences: curvature anisotropy induces anisotropic diffusion; nearly flat regions enhance fluctuation amplitude; curvature concentration suppresses stochastic spreading; and topology-changing events produce discrete entropy increments. In all cases, observable behaviour is governed by the same curvature-controlled action, whose realization depends on the geometric scale and signature. Stochastic irreversibility and quantum propagation are therefore not separate principles, but regime-dependent expressions of a single geometric law.

\bibliographystyle{apsrev4-2}
\bibliography{stochastic_refs}

%\section*{Author Statement on AI Use}
%Portions of the typesetting, formatting assistance, and editorial language refinement were supported using AI tools (ChatGPT, OpenAI). All mathematical derivations, equations, proofs, and scientific conclusions are solely the work of the author. AI tools were not used for generating or checking results.

\end{document}